\date{}
\documentclass[11pt]{article}
\usepackage[all]{xy}

\usepackage{amsmath,amsthm,amsfonts,amssymb,amsopn,amscd}     
\usepackage{color}

\DeclareMathOperator{\Tr}{Tr}

\DeclareMathOperator{\Ind}{Ind}

\def\qed{{\unskip\nobreak\hfil\penalty50
\hskip2em\hbox{}\nobreak\hfil$\square$
\parfillskip=0pt \finalhyphendemerits=0\par}\medskip}

\def\endproof{\null\hfill\qed\endtrivlist\noindent}

\def\tilde{\widetilde}

\def\Ad{{\hbox{\rm Ad}}}

\def\o{{\omega}}

\def\a{\alpha}
\def\b{\beta}

\def\e{\varepsilon}

\def\l{\lambda}

\def\r{\rho}

\def\th{\theta}

\def\Om{\Omega}

\newcommand{\ben}{\begin{equation}}
\newcommand{\een}{\end{equation}}

\def\gA{{\mathfrak A}}

\def\A{{\cal A}}

\def\F{{\cal F}}

\def\M{{\cal M}}
\def\N{{\cal N}}

\def\L{{\mathcal L}}

\def\H{{\cal H}}
\def\K{{\cal K}}
\def\S{{\cal S}}

\def\f{{\varphi}}
\def\s{{\sigma}}

\def\PSL{{{\rm PSL}(2,\mathbb R)}}

\def\S2{S^{1(2)}}

\def\fin{\f_{\rm in}}
\def\fout{\f_{\rm out}}

\newtheorem{theorem}{Theorem}[section]
\newtheorem{lemma}[theorem]{Lemma}

\newtheorem{corollary}[theorem]{Corollary}

\newtheorem{proposition}[theorem]{Proposition}

\theoremstyle{definition} 

\theoremstyle{remark} 

\def\PSL{PSU(1,1)}

\def\RR{{\mathbb R}}

\def\SL2{{{\rm SL}(2,\RR)}}

\def\PSL2{{{\rm PSL}(2,\RR)}}

\def\U1{{{\rm V}(1)}}
\def\SU2{{{\rm SV}(2)}}

\def\SU{{{\rm SU}}}

\def\A{{\mathcal A}}

\def\F{{\mathcal F}}
\def\H{{\mathcal H}}

\def\K{{\mathcal K}}

\def\M{{\mathcal M}}
\def\N{{\mathcal N}}

\title{\Huge{On Landauer's principle and  bound \\ for infinite systems
}}

\author{{\sc Roberto Longo\footnote{Supported by the ERC Advanced Grant 669240 QUEST ``Quantum Algebraic Structures and Models'', MIUR FARE R16X5RB55W  QUEST-NET, GNAMPA-INdAM and Alexander von Humboldt Foundation.}
}
\\
Dipartimento di Matematica,
Universit\`a di Roma Tor Vergata,\\
Via della Ricerca Scientifica, 1, I-00133 Roma, Italy\\
E-mail: {\tt longo@mat.uniroma2.it}
\\
}

\date{}
\begin{document} 

\maketitle

\begin{abstract}
Landauer's principle provides a link between Shannon's information entropy and Clausius' thermodynamical entropy. We set up here a
basic formula for the incremental free energy of a quantum channel, possibly relative to infinite systems, naturally arising by an Operator Algebraic point of view. By the Tomita-Takesaki modular theory, we can indeed describe a canonical evolution associated with a quantum channel state transfer. Such evolution is implemented both by a modular Hamiltonian and a physical Hamiltonian, the latter being determined by its functoriality properties. 
This allows us to make an intrinsic analysis, extending our QFT index formula, but without any a priori given dynamics;
the associated incremental free energy is related to the logarithm of the Jones index and is thus quantised. This leads to a general lower bound for the incremental free energy of an irreversible quantum channel which is half of the Landauer bound, and to further bounds corresponding to the discrete series of the Jones index. In the finite dimensional context, or in the case of DHR charges in QFT, where the dimension is a positive integer, our lower bound agrees with Landauer's bound.   

\end{abstract}

\newpage
\section{Introduction}
\label{sec:intro}
Quantum information is  an increasing lively subject. At the present time, most papers in the subject deal with finite quantum systems, i.e. multi-matrix algebras. This is justified from the quantum computation point of view inasmuch as only few qubits quantum computers have so far been built up. However, from the functional analytic and conceptual point of view, it is natural to look at infinite quantum systems too. On one hand big quantum systems, as quantum computers in perspective or quantum black holes, may be idealised as infinite systems, on the other hand operator algebras not of type $I$ reveal a rich structure that is not directly visible within the finite dimensional context and may thus both offer a new insight and provide effective tools for the analysis. In particular, looking at the subject from the Quantum Field Theory point of view is expected to lead to a new perspective within this framework. 

In this paper we shall see a dynamics that is naturally associated with a quantum channel. At the root of our investigation is our analysis 
in \cite{L97} dealing with the time evolution in black hole thermodynamics,  that also led us to a QFT index formula \cite{L01}.
The dynamics in this paper is intrinsic and originates from the modular structure associated with Connes' bimodules.  Now, this evolution is canonically implemented both by a modular Hamiltonian and by a physical Hamiltonian, the latter being characterised by its good functoriality properties with respect to the bimodule tensor categorical structure. We can compare these two Hamiltonians; in the simplest factor case, one is simply obtained by shifting the other. We then define an intrinsic incremental free energy associated with the quantum channel.  The incremental free energy turns out to be proportional to the logarithm of the Jones index; so, in particular, we obtain a lower bound which is related to the Landauer's bound, and agrees with that in the finite dimensional context. The structure associated with quantum systems with a non trivial classical part, namely non-factor von Neumann algebras, is more involved; we shall deal here with the case where the classical part is finite dimensional.

In Section \ref{Section2} we shall get results on Connes' bimodules and completely positive maps that play a prime role in our analysis and have their own interest.  

We now state some general facts that are at the basis of our investigation.

\subsection{Entropy}
We begin by recalling the basic forms of entropy that are related to our work. 

\subsubsection{Thermodinamical entropy}
The concept of entropy arose in thermodynamics and is due to Clausius. If only reversible processes take place in an isolated, homogeneous system, then the integral over a closed path of the form $\frac{dQ}{T}$ vanishes: $\oint \frac{dQ}{T} =0$. Thus $\frac{dQ}{T}$ is an exact differential form. As heat (energy) $Q$ and temperature $T$ are respectively an extensive and an intensive quantity, there must so exist an extensive state function $S$ such that $dS = \frac{dQ}{T}$. Clausius named  this new appearing quantity entropy from Greek ``transformation content''.

The second principle of thermodynamics asserts that, for an isolated system, the entropy never decreases
\[
dS \geq 0
\]
and indeed $dS = 0$ just for reversible transformations.

\subsubsection
{Information entropy} 
The theory of quantum information promotes classical information theory to the quantum world. 
One may say that a classical system is described by an abelian operator algebra $\M$, while a quantum system by a non commutative operator algebra $\M$ and one firstly aims to describe the information carried by a state of a system $\N$ when transfered to a state of a system $\M$ by a classical/quantum channel.

Let's consider finite systems first. In the classical case, the information of a state $\f$ on a system $\mathbb C^n$ ($n$-point space) with probability distribution $\{p_1,\dots p_n\}$ is measured by the Shannon entropy
\begin{equation}\label{ShE}
S(\f) = - k\sum_i p_i\log p_i \  ,
\end{equation}
where $k$ is a proportionality constant that for the moment may be put equal to 1. The point is that the probability of independent events is multiplicative, while information is additive, thus the information is to be proportional to the logarithm of the probability. So \eqref{ShE} measure the average information carried by the state $\f$. 

In the quantum case $\M$ is a matrix algebra, $\f$ is a normalised, positive linear functional on $\M$ with density matrix $\r$, and the information entropy is given by von Neumann entropy
\begin{equation}\label{vne}
S(\f) = -\Tr(\r\log \r) .
\end{equation}
If the finite system is to encode both classical and quantum information, $\M$ is an arbitrary finite dimensional $C^*$-algebra. 

A {\it quantum channel} between the two finite systems $\N$ and $\M$ is a completely positive map $\a: \N\to \M$ (possibly trace preserving, unital). 
Already within this relatively simple, finite dimensional setting one can see  
remarkable quantum structures and conceptual aspects.

Now, in the classical case, the passage from finite to infinite systems does not provide a new conceptual insight. On the other hand, we will see in this paper  new aspects that  do emerge in this passage from finite to infinite systems in the quantum case.

An infinite quantum system, possibly with a classical part too, will be described by a von Neumann algebra $\M$; the von Neumann entropy of a normal state $\f$ on $\M$ makes no sense in this case, unless $\M$ is of type $I$; however Araki's relative entropy between two faithful normal states $\f$ and $\psi$ on $\M$ is defined in general by
\[
S(\f | \psi) \equiv -(\eta ,  \log\Delta_{\xi , \eta}\,\eta) \ ,
\]
where $\xi,\eta$ are the vector representatives of $\f ,\psi$ in the natural cone $L^2_+(\M)$ and $\Delta_{\xi , \eta}$ is the relative modular operator associated with $\xi, \eta$, see \cite{BR,OP}. Relative entropy is one of the key concepts used in the following.

In this paper a quantum channel will be a normal, unital, completely positive $\a$ map between von Neumann algebras with finite dimensional centers, where  $\a$ has finite index and $\f$ and $\psi$ input and output states with respect to $\a$. The index will be the Jones index of a normal bimodule canonically associated with $\a$. Modular theory and subfactor theory are indeed at the root of our analysis, as we shall later explain.  

\subsubsection
{Statistical mechanics entropy}
According to Boltzmann, for an isolated system in thermal equilibrium, the entropy is associated with the logarithmic counting of the number of all possible states of the system, namely $S = k\log W $, where $W$ is the number of all possible microstates compatible with the given macrostate of the system; here $k$ is the Boltzmann constant. 

We can view $S$ as a measure of our lack of knowledge about our system, thus getting a link between thermodynamical entropy and information entropy. 

Since the equilibrium distribution at inverse temperature $\b = \frac{1}{kT}$ is given by the Gibbs distribution $e^{-\b H}/Z(\b)$ with $H$ the energy and $Z$ the normalising partition function, we see that, say in the quantum case, $S$ is given by formula \eqref{vne}  with density matrix $\r =e^{-\b H}/Z(\b)$.

\subsubsection
{Black hole entropy}
In black hole thermodynamics, the laws of thermodynamics are promoted to a quantum black hole framework. In this context, entropy 
has been subject to several different interpretations, both as statistical mechanics entropy and as information entropy. Fundamental aspects here are provided by the Hawking thermal radiation, and by Bekenstein area law giving entropy a geometrical interpretation as proportional to the area of the black hole event horizon. We refer to the literature for more on this subject, e.g. \cite{W}.

\subsection{Maxwell's demon. Landauer's bound} 
Towards the end of the 19th century, Maxwell suggested a thought experiment to show
how the Second Law of Thermodynamics might hypothetically be violated, the so called
Maxwell's demon experiment, that we recall. 

A gas is in equilibrium in a box and a wall is put to divide the box in two halves $A$ and $B$. 
A little being, the demon, controls a tiny door on the wall between $A$ and $B$. 
As an individual gas molecule approaches the door in $A$ or $B$, the demon quickly opens and shuts the door in order to allow only the faster molecules to pass from $A$ to $B$, and only the slower molecules to pass from $B$ to $A$. The average molecule speed thus increases in $B$, and decreases in $A$. Since faster molecules give rise to higher temperature, as a result the temperature in $B$ becomes higher than in $A$, so the entropy of the system decreases, thus violating the Second Law of Thermodynamics. Indeed a thermodynamical engine could extract work from this temperature difference.

Maxwell demon experiment, and its subsequent more refined versions, have long been a matter of debate in the Physics community, see e.g. \cite{LeRe}. An important contribution came by Szilard with a further idealisation where the gas has only one molecule.  
Szilard pointed out that the act of Maxwell's demon to measuring molecular speed would require an expenditure of energy. So one must consider the entropy of the total system including the demon. The expenditure of energy by the demon would produce an increase of the entropy of the demon, which would be larger than the decrease of the entropy of the gas.

Rolf Landauer \cite{La} realised however that some measuring processes need not increase thermodynamic entropy as long as they were thermodynamically reversible. Landauer argued  that {\it information is physical} and this principle was central to solving the paradox of Maxwell's demon. Bennett \cite{B} noted that the demon has to memorise the information he acquires about the gas molecules. He argued that after a full cycle of information the demon's memory has to be reset to its initial state to allow for a new iteration. According to Landauer's principle, the erasure process will always dissipate more entropy than the demon annihilated during one cycle, in full agreement with the second law of thermodynamics. Taking into account Shannon's information entropy, a generalised second law of thermodynamics holds true by considering the total system  and the total (thermodynamical $+$ information) entropy.

{\it Landauer's principle} states that ``any logically irreversible manipulation of information, such as the erasure of a bit or the merging of two computation paths, must be accompanied by a corresponding entropy increase in non-information bearing degrees of freedom of the information processing apparatus or its environment'' \cite{B}.

Another way of phrasing this principle is that if an observer loses information about a physical system, the observer loses the ability to extract work from that system.

Landauer's principle sets a lower bound of energy consumption of computation or logical operation, also known as {\it Landauer's bound}. 
For an environment at temperature $T$, energy $E = ST$ must be emitted into that environment if the amount of added entropy is $S$. For a computational operation in which 1 bit of logical information is lost, the amount of entropy generated is at least $k \log 2$ and so the energy that must eventually be emitted to the environment is 
\[
E \geq kT \log 2\ ,
\]
with $k$ the Boltzmann constant. If no information is erased, a thermodynamically reversible logical operation is theoretically possible with no release of energy $E$, therefore the above bound concerns irreversible transformations. 

Recently, physical experiments have tested Landauer's principle and confirmed its predictions, see \cite{BAP}.

\subsection{Underlying mathematical and physical context}
We now recall a few facts that play a particular role in our paper.

\subsubsection{Modular theory, the intrinsic dynamics}
As is well known, a {\it von Neumann algebra} $\M$ generalises at the same time the notion of (multi-)matrix algebra and the one of measure space: 
if $\M$ is finite dimensional then $\M$ is direct sum matrix algebras, if $\M$ is abelian then $\M$ can be identified with $L^\infty(X,\mu)$ for some measure space $(X,\mu)$, so $\M$ is sometimes called a ``noncommutative measure space". 

Let $\f$ be a faithful normal state of $\M$ (noncommutative integral). The Tomita-Takesaki modular theory provides a canonical, {\it intrinsic dynamics} associated with $\f$: a one-parameter group of automorphism  $\s^\f$  of $\M$,
\[
t\in\mathbb R \longmapsto \s^\f_t \in {\rm Aut}(\M) \ ,
\]
called the {\it modular group} of $\f$. Among its remarkable properties, we mention here the following:

$\bullet$ $\s^\f$ is a purely noncommutative object: $\s^\f$ acts identically if and only if $\f$ is tracial, therefore $\s^\f$ is not visible within the classical abelian case. 

$\bullet$ $\s^\f$ does not depend on $\f$ up to inner automorphisms by Connes' Radon-Nikodym theorem \cite{C73}; in particular, $\s^\f$ can be an outer action  only in the infinite dimensional case, indeed in the type $III$ case.

$\bullet$ $\s^\f$ is characterised by the KMS thermal equilibrium condition at inverse temperature $\b = -1$ with respect to the state 
$\f$. The KMS condition appears in Quantum Statistical Mechanics, see Section \ref{QRSM}; thus modular theory is directly connected with Physics.

By the last point, if we have a physical evolution satisfying the KMS condition at inverse temperature $\b >0$ w.r.t. a state $\f$, then we may identify this evolution with the rescaled modular group $t\mapsto \s^\f_{-t/\b}$, yet the modular evolution exists independently of any underlying physical setting. 
\subsubsection{Jones' index}\label{JI}
Let $\N\subset\M$ be an inclusion of factors. The {\it Jones index} $[\M:\N]$ measures the relative size of $\N$ in $\M$. It was originally defined for factors with a tracial state \cite{Jo}, then extended to arbitrary inclusions of factors with a normal conditional expectation $\e:\M\to\N$ by Kosaki \cite{K86}, and then in \cite{L89}. 
One of the main properties of the index is that its values are quantised:
\begin{equation}\label{iv}
[\M:\N] = 4\cos^2\frac{\pi}{k}\ , \ k = 3, 4, \dots\quad {\rm if} \ [\M:\N]<4 \   ,
\end{equation}
by Jones' theorem \cite{Jo}. 

In general, the well behaved index to be considered is the one with respect to the minimal expectation, the minimum over all possible indices. 
The multiplicativity of the minimal index is shown in \cite{KL,L92} and refs therein. 

Jones' index appears in many contexts in Mathematics and in Physics. In \cite{L89} we showed the {\it index-statistics relation}
\[
\text{DHR dimension} = \sqrt{\text{Jones index}}\ ,
\]
 between the Doplicher-Haag-Roberts statistical dimension of a superselection sector in Quantum Field Theory and the index of a localised endomorphism that represents the sector. Furthermore, our formula  in QFT \cite{L97} provides an interpretation of the logarithm of the index from the entropy viewpoint, see also Sections \ref{SectKW} and \ref{FC} here. 
 
In this paper, we are going to consider possibly infinite physical systems with a finite classical part, so
von Neumann algebras $\N, \M$ with finite dimensional centers. If $\N\subset\M$, the minimal expectation and the minimal index can then be defined, and the latter is a scalar subject to the above restriction \eqref{iv} if the centers of $\N$ and $\M$ intersect trivially, see \cite{Hav,FI} and refs therein. 

In a forthcoming paper \cite{GL}, we shall see how to define a non-scalar dimension
that is multiplicative in this general framework. This fits into the definition of the physical Hamiltonian we shall give below. Here, in order to focus more on the quantum information side, we do not dwell on non-scalar dimension issues; yet this paper is self-contained.   
\subsubsection{Quantum relativistic statistical mechanics}\label{QRSM}
Let $\M$ be a finite purely quantum system, namely $\M$ is the algebra on $n\times n$ complex matrices. The time evolution $\tau$ is a one-parameter automorphism group of $\M$ implemented by a one-parameter group $U$ of unitaries in $\M$, namely  $\tau_t(X) \equiv U(t)XU(-t)$, $X\in\M$. We have
$U(t) = e^{itH}$ where $H\in \M$ is a positive selfadjoint operator, the Hamiltonian. 

As is well known, a state $\f$ of $\M$ in thermal equilibrium at inverse temperature $\b >0$ is characterised by the Gibbs condition, namely
\[
\f(X) = \Tr(\r X)\ ,
\]
where $\r$ is the density matrix $\rho = e^{-\beta H}/\Tr(e^{-\b H})$.

At infinite volume, $\M$ becomes an infinite dimensional operator algebra, in general no trace $\Tr$ exists any longer on $\M$ and the evolution $\tau$ is not inner. A state $\f$ in thermal equilibrium at inverse temperature $\b >0$ is now characterised by the {\it KMS condition}: 
\smallskip

\noindent
For every $X,Y\in \M$, there is a function $F_{XY}\in A(S_\b)$ such that 
\begin{equation}\label{KMS}
\begin{aligned}
\mkern-328mu (a)\quad & F_{XY}(t)= \omega\big(X\tau_t(Y)\big)\  , \\
\mkern-328mu (b)\quad & F_{XY}(t+i\b)=\omega\big(\tau_t(Y)X\big)\ ,
\end{aligned}
\end{equation}
where $A(S_\b)$ is the algebra of functions analytic in the strip
$S_\b=\{0<\Im z <\b\}$, bounded and continuous on the closure $\bar
S_\b$.  
\smallskip

\noindent
KMS states generalise Gibbs states in the context of infinite systems, see  \cite{H,T,BR}.

As said, if $\M$ is a von Neumann algebra and $\f$ a faithful normal state of $\M$, then $\tau$ can be identified with the shifted modular group $\s^\f$ of $\f$. For an evolution on a $C^*$-algebra $\gA$ with KMS state $\f$, one gets into the von Neumann algebra framework  by considering the weak closure $\M$ of $\gA$ in the GNS representation associated with $\f$, the KMS condition holds then on $\M$ too.

In quantum relativistic statistical mechanics, locality and relativistic invariance have both to hold. In essence, as far as we are concerned here, quantum relativistic statistical mechanics is the study of KMS states in Quantum Field Theory.

\subsubsection
{The analog of the Kac-Wakimoto formula and a QFT index theorem}\label{SectKW}
A local conformal net $\A$ on $S^1$ is the operator algebraic framework to study chiral Conformal Quantum Field Theory, see \cite{KL04}. If $L_0$ is the conformal Hamiltonian (generator of the rotation one parameter unitary group) in the vacuum representation, and $L_\r$ is the conformal Hamiltonian in any representation $\r$ of $\A$, one expects the following formula to hold:
\begin{equation}\label{KW}
\lim_{t\to 0^+} \frac{\Tr(e^{-tL_\r})}{\Tr(e^{-tL_0})} = d(\r)\ ,
\end{equation}
here $d(\r)$ is the dimension of $\r$ (see Section \ref{JI} and \cite{LR}). 

Although formula \eqref{KW} has been checked for most models, see \cite{X, KL05}, it stands up as an important unproven conjecture. 

In \cite{L97}, we have however proven an analog of formula \eqref{KW} that holds true in full generality, where the rotation flow is replaced by a geometric KMS flow 
in QFT, for instance the dilation flow in CFT or the boost flows in general Quantum Field Theory \cite{BW,HL}.  To stay in a specific context, let $H_0$ and $H_\r$ be the Hamiltonian for a uniformly accelerated observer in the Minkowski spacetime with acceleration $a>0$, equivalently $H_0$ and $H_\r$ are the generators the geodesic flow evolution in the Rindler spacetime, respectively in the vacuum representation and in the representation $\r$; then
\begin{equation}\label{AKW}
(\Om, e^{-t H_\r} \Om)\big\vert_{t= \b} = d(\r)\ ,
\end{equation}
with $\Om$ the vacuum vector and $\b = \frac{2\pi}{a}$ the inverse Hawking-Unruh temperature, see also Section \ref{FC}.

The proof of formula \eqref{AKW} is based on a tensor categorical analysis and has been subsequently extended as a ground for a QFT index theorem \cite{L01}. 

In this paper, in particular, we generalise formula \eqref{AKW} without any reference to a given KMS physical flow, we indeed rely on an intrinsic modular evolution so to set up a universal formula. Of course, this formula gives back \eqref{AKW} and other formulas when applied to specific physical settings. Because of its generality, it can now be applied also to the Quantum Information setting, the context we are mainly going to consider in the following. 

\section{Bimodules and completely positive maps}
\label{Section2}

In order to simplify our exposition, we assume all von Neumann algebras in this paper to have separable preduals, namely to be representable on a separable Hilbert space. And all Hilbert spaces will be separable.

\subsection{Connes' bimodules and the intrinsic dynamics}
\label{CB}
We start by recalling the basic facts on normal bimodules over von Neumann algebras (see \cite{C, Sau83, Po87, L90, Fa})
and develop further material. 

Given von Neumann algebras $\N$ and $\M$, by an $\N - \M$ {\it bimodule} $\H$ we mean a Hilbert space $\H$ equipped with a normal left action of $\N$ and a normal right action of $\M$. Namely we have a normal representation $\ell = \ell_\H$ of $\N$ on $\H$ and a normal anti-representation $r = r_\H$  of $\M$ on $\H$ (thus $r(m_1 m_2) = r(m_2)r(m_1)$\,\!) such that 
$\ell(\N)$ and $r(\M)$ mutually commute. The natural notation
\[
n\xi m \equiv \ell(n)r(m)\xi\, , \quad n\in\N, \ m\in\M,\ \xi\in\H \, ,
\]
will be reserved for the identity bimodule below. 

Let $\M^o$ be the von Neumann algebra opposite to $\M$, with $m\mapsto m^o$ the natural anti-isomorphism of $\M$ with $\M^o$. 
Since an anti-representation $r$ of $\M$ corresponds to a representation $r^o$ of $\M^o$, i.e. $r^o(m^o) = r(m)$, an $\N - \M$ {\it bimodule} $\H$ corresponds to a binormal representation $\pi$ of $\N\odot\M^o$, the algebraic tensor product of $\N$ and $\M^o$, namely $\pi$ is a representation of 
$\N\odot\M^o$ on $\H$ whose restriction both on $\N\odot 1$ and $1\odot \M^o$ is normal, indeed 
$\pi(n\odot m^0) = \ell(n) r(m)$. The representation $\pi$ extends and corresponds uniquely to a representation of the maximal tensor product $C^*$-algebran $\N\bigotimes_{\rm max} \M^o$, the completion of $\N\odot\M^o$ with respect to the maximal $C^*$-norm. 

So there are natural notions of {\it direct sum} of bimodules and {\it intertwiner} between bimodules: they are the ones that appear when we view an $\N -\M$ bimodule as a representation of $\N\bigotimes_{\rm max} \M^o$, indeed bimodules form a $C^*$-category.

We shall denote by $L^2(\M)$ the {\it identity} $\M -\M$ bimodule, which is unique (up to unitary equivalence). If $\M$ acts on a Hilbert space $\H$ with cyclic and separating vector $\xi$, then $\H = L^2(\M)$ with actions $\ell(m)\eta = m\eta$, $r(m)\eta = Jm^* J\eta$, $m\in\M$, $\eta\in\H$, where $J$ is the modular conjugation of $\M$ associated with $\xi$, indeed with any cyclic (thus separating) vector in the natural positive cone $L^2(\M)_+$ given by $\xi$. 

The {\it conjugate} $\bar\H$ of the $\N -\M$ bimodule $\H$ is the $\M -\N$ bimodule over the conjugate Hilbert space $\bar\H$ with actions 
\begin{equation}\label{conj}
\ell_{\bar\H}(m) \bar\xi = \overline{r_\H(m^*)\xi} \ , \quad r_{\bar\H}(n)\bar\xi = \overline{\ell_\H(n^*)\xi}\ .
\end{equation}
If $\K$ is an $\N -\M$ bimodule and $T:\H\to \K$ an intertwiner, the conjugate interwiner $\bar T : \bar\H \to \bar\K$ is defined by
\[
\bar T \bar\xi = \overline{T\xi}, \quad \xi\in\H \ .
\]
The map $T\to \bar T$ is anti-linear, multiplicative and commutes with the $^*$-operation.  (Actually, $\bar T$ is defined for any linear operator $\H\to\K$).

Moreover, there exists an {\it (internal) tensor product} of bimodules and of intertwiners (\cite{C, Sau83}): if $\H$ is a $\M_1 -\M_2$ bimodule and 
$\K$ is a $\M_2 -\M_3$ bimodule, then $\H\otimes\K$ is a $\M_1 -\M_3$ bimodule. We shall give the definition of the tensor product, up to unitary equivalence, later below. 

Note that we have also the operation of {\it external tensor product}. Let $\H_k$ be an $\N_k - \M_k$ bimodule, $k =1,2$. The external tensor product $\H_1 \bigotimes \H_2$ is the obvious $\N_1\bigotimes\N_2 - \M_1\bigotimes\M_2$ bimodule on the Hilbert space $\H_1\bigotimes\H_2$. 
We shall use the symbol $\otimes$ for the internal tensor product, and the larger $\bigotimes$ for the usual external tensor product. 

Let now $\N$ and $\M$ be von Neumann algebras with finite dimensional centers and $\H$ a  $\N - \M$ bimodule $\H$. We shall say that $\H$ has finite index if the inclusion $\ell(\N)\subset r(\M)'$ has finite index (\cite{K86,Hav, FI}). In this case, if the centers are finite dimensional, the {\it index} of $\H$ is defined by
\[
\Ind(\H) \equiv [r(\M)' : \ell(\N)] = [\ell(\N)' : r(\M)]\ ,
\]
here the square brackets denote the {\it minimal index} of an inclusion of von Neumann algebras with finite dimensional centers (see \cite{L90,FI} and references therein), we refer to \cite{GL} for more on the index structure we need here.

We shall say that $\H$ is {\it connected} if 
\[
\ell(\N)\cap r(\M) = \ell\big(Z(\N)\big)\cap r\big(Z(\M)\big) = \mathbb C \ .
\]
Here $Z(\cdot)$ denotes the center. If $\H$ has finite index, then $\H$
is the direct sum of finitely many connected bimodules. For all results in this paper, we may deal with connected bimodules only, the general case being an immediate consequence by considering the direct sum along the atoms of  $\ell\big(Z(\N)\big)\cap r\big(Z(\M)\big) $. 

Note that $\Ind(\H)$ is an element of $\ell(\N)\cap r(\M)$ and $\Ind(\H)\geq 1$, in particular the index is a scalar if $\H$ is connected.

Let $\Ind(\H) < \infty$ and $\H$ be connected as above.
The scalar {\it  dimension} of $\H$ is the square root of the index
\ben\label{id}
d_\H \equiv \sqrt{\Ind(\H)} \ .
\een
If $\H$ is not connected, we take the minimal projections $e_k$ of $\ell(\N)\cap r(\M)$ and set $d_\H \equiv \sum_k d_{(e_k\H)}\, e_k$. The equality \eqref{id} still holds. 

Given faithful, normal, positive linear functional $\f,\psi$ on $\N$ and $\M$ (usually states), we define the {\it modular operator} $\Delta_\H(\f |\psi)$  of $\H$ with respect to $\f,\psi$ as
\begin{equation}\label{modHam} 
\Delta_\H(\f |\psi) \equiv d(\f\cdot\ell^{-1}\cdot\e)\big/ d(\psi\cdot r^{-1}) \ ,
\end{equation}
where the right hand side is Connes' spatial derivative \cite{C80} for the pair $r(\M)' , r(\M)$ w.r.t. the states $\f\cdot\ell^{-1}\cdot\e$ and $\psi\cdot r^{-1}$
and $\e: r(\M)' \to \ell(\N)$ is the minimal conditional expectation, see also \cite{St,Y}.

In \eqref{modHam} the meaning of $\ell^{-1}$ and $r^{-1}$ is clear if
$\ell$ and $r$ are injective maps, which is automatic if $\N$ and $\M$ are factors. 
In general, let $p$ the central support of $\ell$ (i.e. $1 - p$ is the largest projection in $Z(\N)$ in the kernel of $\ell$) and similarly $q$ the central support of $r$. Then $\ell_p \equiv\ell|_{\N_p}$ and $r_q \equiv r|_{\M_q}$ are injective maps and we put $\ell^{-1} \equiv \ell^{-1}_p$, $r^{-1} \equiv r^{-1}_q$ (thus $\f\cdot \ell^{-1} \big(\ell(n)\big) = \f(np)$ and $\psi\cdot r^{-1} \big(r(m)\big) = \psi(mq)$\,). 

The following key lemma follows by combining the properties of the spatial derivative \cite{C80} and Takesaki's theorem on conditional expectations \cite{T}. 
\begin{lemma}\label{lemcov} 
For every $t\in\mathbb R$, we have
\begin{equation}\label{modH}
\Delta^{it}_\H(\f|\psi)\ell(n)\Delta^{-it}_\H(\f|\psi) = \ell\big(\s^\f_t(n)\big)\ , \quad 
\Delta^{it}_\H(\f|\psi)r(m)\Delta^{-it}_\H(\f|\psi) = r\big(\s^\psi_{t}(m)\big)\ ,
\end{equation}
$n\in\N$, $m\in\M$. 
\end{lemma}
\begin{proof}
$\big(d(\f\cdot\ell^{-1}\cdot\e)\big/ d(\psi\cdot r^{-1})\big)^{it}$ implements $\s_t^{\f\cdot\ell^{-1}\cdot\e}$ on $r(\M)'$ and 
$\s_{-t}^{\psi\cdot r^{-1}}$ on $r(\M)$. 
By Takesaki's theorem $\s_t^{\f\cdot\ell^{-1}\cdot\e}$ restricts to $\s_t^{\f\cdot\ell^{-1}}$ on $\ell(\N)$. Since 
$\s_t^{\f\cdot\ell^{-1}}\! \big(\ell(n)\big) = \ell\big(\s_t^\f(n)\big)$ and 
$\s_{-t}^{\psi\cdot r^{-1}}\! \big(r(m)\big) = r\big(\s_{t}^\psi(m)\big)$
we have the equalities in \eqref{modH}.
\end{proof}
Note also that
\ben\label{frc}
\Delta^{it}_\H(\f|\psi)X\Delta^{-it}_\H(\f|\psi) = X\ ,\quad X\in \big(\ell(\N)\vee r(\M)\big)' \ ,
\een
because the minimal conditional expectation is tracial on the relative commutant. 

Since
\[
\Delta_\H(\l\f |\psi) = \l\Delta_\H(\f |\psi),\quad \Delta_\H(\f |\l\psi) = \l^{-1}\Delta_\H(\f |\psi), \quad \l >0,
\]
we may normalise $\f$ and $\psi$ and study the case where they are states, although non normalised functionals may appear when considering reduced bimodules as before Lemma \ref{lemcov}. 

The definition of $\Delta_\H(\f|\psi)$ is not symmetric. If we define the right modular operator $\Delta'_\H(\f|\psi) = d(\f\cdot\ell^{-1})\big/ d(\psi\cdot r^{-1}\cdot\e)$, with $\e: \ell(\N)' \to r(\M)$ the minimal conditional expectation,
we have by Kosaki's formula \cite{K86} that
\begin{equation}\label{Ko}
\Delta'_\H(\f|\psi) = \Ind(\H)\, \Delta_\H(\f|\psi) \ .
\end{equation}
We call $\log \Delta_\H(\f |\psi)$ the {\it modular Hamiltonian} of $\H$ with respect to the states $\f$ and $\psi$. 
\smallskip

Let $\N,\M$ be von Neumann algebras and $\H$ an $\N-\M$ bimodule. Let $\F$ be a type $I_\infty$ factor, namely $\F$ is isomorphic to $B(\K)$ with $\K$ a separable, infinite-dimensional Hilbert space, and consider the identity $\F-\F$ bimodule $L^2(\F)$ ($B(\K)$ acting on Hilbert-Schmidt operators by left and right multiplication).  Then the external tensor product $\H\bigotimes L^2(\F)$ is naturally an $\N\bigotimes\F -\M\bigotimes\F$ bimodule, the ``ampliation'' of $\H$. 
As $\N\bigotimes\F$, $\M\bigotimes\F$ are properly infinite. By considering the ampliation,  most proofs can be easily reduced to the case of bimodules over properly infinite von Neumann algebras.

We begin to analyse the functoriality properties of the modular Hamiltonian in the factorial case. 

The internal tensor product $\H\otimes\K$ of bimodules $\H$ and $\K$ was defined in \cite{C, Sau83}, we shall explain it in the next section.  Notice that the modular unitary $\Delta^{it}_\H(\f_1|\f_2)$, $t\in\mathbb R$,  is an intertwiner $\H\to \H_t$ (with $\H_t \equiv {}_{\s^{\f_1}_t}\H_{\s^{\f_2}_t}$) and similarly for
$\Delta^{it}_{\K}(\f_2|\f_3)$ so, in particular, the notions of intertwiner tensor product and conjugation are defined for them.

Let's say that an $\N-\M$ bimodule $\H$ is {\it factorial} if both $\ell_\H(\N)$ and $r_\H(\M)$ are factors.
\begin{theorem}\label{th1}
Let $\M_k$ be von Neumann algebras and $\f_k$ faithful normal positive linear functionals of $\M_k$, $k = 1,2,3$. With $\H,\H'$ finite index, factorial $\M_1 - \M_2$ bimodules and $\K$ a  finite index, factorial $\M_2 - \M_3$ bimodule, we have:
\begin{itemize}
\item[$(a)$]  $\Delta^{it}_\H(\f_1|\f_2) \otimes \Delta^{it}_{\K}(\f_2|\f_3) = \Delta_{\H\otimes\K}^{it}(\f_1|\f_3)$;
\item[$(b)$]  $\Delta^{it}_{\bar\H}(\f_2|\f_1) = \Ind(\H)^{-it}\, \overline{\Delta^{it}_{\H}(\f_1|\f_2)}$;
\end{itemize}
If $T:\H\to\H'$ is a bimodule intertwiner, then
\begin{itemize}
\item[$(c)$] $T\Delta^{it}_\H(\f_1|\f_2) = (d_{\H'}/d_\H)^{it} \Delta^{it}_{\H'}(\f_1|\f_2)T$;  
\end{itemize}
Here, in $(a)$, the tensor product is taken with respect to the state $\f_2$.
\end{theorem}
\medskip

\noindent
The proof of Theorem \ref{th1} is postponed to the Appendix \ref{appendix}, see however the comment after Proposition \ref{intert} below. 
\medskip

\noindent
It follows from $(c)$ of Theorem \ref{th1} that if $\H$ is factorial and direct sum of finitely many $\M_1 - \M_2$ bimodules $\H_k$, i.e. $\H = \bigoplus_k \H_k$, then
\[
\log \Delta_{\H}(\f_1|\f_2) + \log d_\H = \bigoplus_k \big(\log \Delta_{\H_k}(\f_1|\f_2) + \log d_{\H_k}\big) \ .
\]
By Theorem \ref{th1}, in the factorial case it is natural to consider the unitary one parameter group $d_\H^{it}\Delta^{it}_{\H}(\f_1|\f_2)$ in order to have good functoriality properties. 

Theorem \ref{th1} does not directly extends  to the non-factorial case. For a finite-index inclusion of von Neumann algebras with finite dimensional centers, there exists an associated, unique minimal conditional expectation; however the composition of minimal expectations may fail to be minimal. In other words, the multiplicativity of the minimal index does not always hold in the non-factor case.

If $\H$ is a finite-index $\N-\M$-bimodule $\H$ with finite dimensional centers $Z(\N)$ and $Z(\M)$, then $\H$ is direct sum of factorial $\N-\M$ bimodules, namely $\H = \bigoplus_{i,j} \H_{ij}$, with $\H_{ij} = \ell_\H(p_i)r_\H(q_j)\H$ the reduced bimodules by the atoms $p_i$ and $q_j$ of $Z(\N)$ and $Z(\M)$ respectively. We refer to the {\it canonical central decomposition} of $\H$ for such a decomposition.

As explained in \cite{GL}, for such an $\N-\M$ bimodule $\H$,
we have to consider the {\it matrix dimension} $D_\H \equiv \{d_{ij}\}$ given by the dimensions 
$d_{ij} \equiv d_{\H_{ij}}$ of the factorial summands $\H_{ij}$. 
If $\H$ is connected, the  scalar dimension is then given by 
\[
d_\H = ||D_\H ||
\]
(Hilbert space linear operator norm of the matrix $D_\H$). 
$D_\H$ is multiplicative,  $d_\H$ is only sub-multiplicative.  

We consider here the strictly positive linear map $D_\H : \H\to\H$ associated with $D_\H$ (still denoted by $D_\H$) given by $D_\H |_{\H_{ij}} = d_{ij}$, so
$D_\H = \sum_{i,j}d_{ij}\ell_\H(p_i)r_\H(q_j)$ belongs to $\ell(\N)'\cap r(\M)'$. 

Now, since the matrix dimension is multiplicative,  
$(a)$ in Theorem \ref{th1} should be generalised in the non-factorial case to
\begin{equation}\label{funct}
\Delta^{it}_\H(\f_1|\f_2)D_\H^{it} \otimes \Delta^{it}_{\K}(\f_2|\f_3)D_\K ^{it} =  \Delta_{\H\otimes \K}^{it}(\f_1|\f_3)D_{\H\otimes \K}^{it}\ ;
\end{equation}  
notice here that $D_\H^{it}$ is an intertwiner $\H\to\H$ and  $\Delta^{it}_\H(\f_1|\f_2)$ is an intertwiner $\H\to\H_t$, so $\Delta^{it}_\H(\f_1|\f_2)D_\H^{it}$ is an intertwiner $\H\to\H_t$, and similarly $\Delta^{it}_{\K}(\f_2|\f_3)D_\K ^{it}$ is an intertwiner $\K\to\K_t$, therefore the internal tensor product in \eqref{funct} is defined. 

In view of Theorem \ref{th1} and the above comments, we thus define the {\it physical Hamiltonian} $K_\H(\f|\psi)$ of $\H$ with respect to the normal faithful states $\f$ and $\psi$ by shifting the modular Hamiltonian in \eqref{modHam}:
\[
K_\H(\f|\psi) \equiv \log\Delta_\H(\f|\psi) + \log D_\H \ .
\]
The name `physical Hamiltonian' is justified by the following Theorem \ref{th2}. 
Indeed, both the modular Hamiltonian and the physical Hamiltonian generate one parameter unitary groups that implement the modular flow and naturally behave under the relative tensor product (the extensive property of energy), yet only the physical Hamiltonian  naturally restricts to sub-bimodules and is symmetric under conjugation (charge conjugation symmetry is a general property in Quantum Physics, see \cite{GL92}). Of course, in the QFT context, our physical Hamiltonian here coincides with the KMS physical Hamiltonian there \cite{L97}. We shall discuss the characterisation of the physical Hamiltonian after Theorem \ref{th2}. 

Alternatively, the definition of the physical Hamiltonian can be given as follows. 
If $\H$ is factorial, we put
\[
K_\H(\f|\psi) = \log\Delta_\H(\f|\psi) + \log d_\H \ .
\]
If $\H$ is an $\N-\M$ bimodule with finite dimensional centers $Z(\N)$ and $Z(\M)$, we consider the canonical central decomposition $\H = \bigoplus_{i,j} \H_{ij}$ as above and put
\[
K_\H(\f|\psi) = \bigoplus_{i,j} K_{\H_{ij}}(\f |\psi) \ ,
\]
with 
$
K_{\H_{ij}}(\f |\psi) = \log\Delta_{\H_{ij}}(\f|\psi)  + d_{\H_{ij}}.
$

The one parameter unitary group 
\[
U_t^\H(\f|\psi) \equiv \Delta_\H^{it}(\f |\psi) D_\H^{it}
\] 
generated by $K_\H(\f|\psi)$
naturally transforms under the tensor categorical operations, in particular the equation
$U_t^{\bar\H}(\psi|\f) = \overline{U_{t}^\H(\f|\psi)}$ fixes the shifting of the modular Hamiltonian if $\H$ is irreducible.  
We now state the main properties of $U_\H$ in the following theorem. 
\begin{theorem}\label{th2}
Let $\M_k$ be von Neumann algebras with finite-dimensional centers and $\f_k$ faithful normal states of $\M_k$, $k = 1,2,3$. With $\H,\H'$ finite index $\M_1 - \M_2$ bimodules and $\K$ a finite index $\M_2 - \M_3$ bimodule, $U^\H(\f_1 |\f_2)$ is a  one parameter unitary group on $\H$ that
depends naturally on $\H,\f_1,\f_2$. 

Namely $U^\H(\f_1 |\f_2)$ {\rm implements the modular dynamics}:
\ben\label{modU}
\begin{aligned}
U_t^\H(\f_1|\f_2)\ell_\H(m_1)U_{-t}^\H(\f_1|\f_2) &= \ell_\H\big(\s^{\f_1}_t(m_1)\big), \quad m_1\in\M_1\, , \\
U_t^\H(\f_1|\f_2)r_\H(m_2)U_{-t}^\H(\f_1 |\f_2) &= r_\H\big(\s^{\f_2}_{t}(m_2)\big), \quad m_2\in\M_2\, ,
\end{aligned}
\een
and the following hold (with $\H\otimes\K \equiv \H\otimes_{\f_2}\K$):
\begin{itemize}
\item[$(a)$:]   $U_t^{\H\otimes\K}(\f_1|\f_3) =  U_t^\H(\f_1|\f_2) \otimes U^\K_t(\f_2|\f_3)$\ \ {\rm (additivity of the energy)};
\item[$(b)$:]  $U_t^{\bar\H}(\f_2|\f_1) =  \overline{U_{t}^{\H}(\f_1|\f_2)}$\ \ {\rm (conjugation symmetry)};
\end{itemize}
If $T:\H\to\H'$ is a bimodule intertwiner
\begin{itemize}
\item[$(c)$:] $TU_t^\H(\f_1|\f_2) = U_t^{\H'}(\f_1|\f_2)T$ \ \ {\rm (functoriality)}.  
\end{itemize}
\end{theorem}
\noindent
Notice that, in particular, property $(c)$ above gives the additivity of the physical unitary evolution under bimodule direct sum:
\ben\label{addsum}
U^{\H\oplus\H'}(\f_1|\f_2) = U^\H(\f_1|\f_2) \oplus U^{\H'} (\f_1|\f_2)\ . 
\een 
\begin{proof}
The first statement \eqref{modU} follows directly from \eqref{modH}.

$(a)$: Let $\H =\bigoplus_{i,k}\H_{ik}$ and $\K =\bigoplus_{k,j}\K_{kj}$ be canonical central decompositions. Then the relative tensor product w.r.t. $\f_2$ is decomposed as
\[
\H\otimes\K = \bigoplus_{i,j,k}\H_{ik}\otimes\K_{kj}
\]
and we apply $(a)$ of Theorem \ref{th1}:
\begin{multline*}
U^\H(\f_1|\f_2) \otimes U^\K(\f_2|\f_3) = \Big(\bigoplus_{i,k}U^{\H_{ik}}(\f_1|\f_2)\Big) \otimes \Big(\bigoplus_{k,j}U^{\K_{kj}}(\f_2|\f_3) \Big)\\
= \bigoplus_{i,j,k}U^{\H_{ik}}(\f_1|\f_2)\otimes U^{\K_{kj}}(\f_2|\f_3)
= \bigoplus_{i,j,k}U^{\H_{ik}\otimes \K_{kj}}(\f_1|\f_3)\\
= \bigoplus_{i,j}U^{(\H\otimes \K)_{ij}}(\f_1|\f_3)
= U^{\H\otimes\K}(\f_1|\f_3) \ .
\end{multline*}

$(b)$: This is satisfied in the factorial case because of $(b)$ in Theorem \ref{th1}, hence in more generality
for bimodules over von Neumann algebras with finite dimensional centers because the conjugation operation commutes with the operation of taking direct sums.

$(c)$: As in the proof of $(c)$ of Theorem \ref{th1}, it is enough to check the additivity w.r.t. direct sum \eqref{addsum}. 
This holds true if $\H\oplus\H'$ is factorial by Theorem \ref{th1}, and also holds true if the direct sum is relative to a central projection of $\M_1$ or of $\M_2$ by definition. Hence $(c)$ always holds because every direct sum is obtained by iteration of these two cases: let $\H = \oplus_{i,j}\H_{ij}$ and $\H' = \oplus_{i,j}\H'_{ij}$ be canonical central decompositions; 
thus $\H\oplus\H' = \oplus_{i,j}(\H_{ij}\oplus\H'_{ij}) = \oplus_{i,j}(\H\oplus\H')_{ij}$ is the canonical central decomposition of $\H\oplus\H'$. Then
\begin{multline*}
U^\H \oplus U^{\H'} = U^{\oplus_{ij}\H_{ij}}\oplus U^{\oplus_{i,j}\H'_{ij}}  = \bigoplus_{i,j}U^{\H_{ij}}\oplus\bigoplus_{i,j} U^{\H'_{ij}}
 = \bigoplus_{i,j}(U^{\H_{ij}}\oplus U^{\H'_{ij}}) 
 \\
 = \bigoplus_{i,j}U^{\H_{ij}\oplus \H'_{ij}} 
  = \bigoplus_{i,j}U^{({\H\oplus \H}')_{ij}} 
 = U^{\H\oplus\H'} \ .
\end{multline*}
\end{proof}
Of course, eq. \eqref{modU} is equivalent to the requirement that Ad$U^{\H}(\f|\psi)$ satisfies the KMS condition on $\ell_\H(\M_1)$ and on $r_\H(\M_2)$. 

We shall further examine the modular and the physical Hamiltonians in Section \ref{QC}. In a first instance, the reader might prefer to use only the scalar dimension and deal here just with the factorial case.

It remains to discuss the uniqueness of $U^\H$. Suppose $V^\H(\f_1|\f_2)$ is another one parameter unitary group on $\H$ that depends naturally on $\H, \f_1 , \f_2$ as in Theorem \ref{th2}. By property $(c)$, we have $V_t^\H(\f_1|\f_2) = U_t^\H(\f_1|\f_2)z_\H(t)$ with $z_\H\equiv z_\H(\f_1|\f_2)$ a one parameter unitary group on $\H$, and by property \eqref{modU} $z_\H(t)\in\ell_\H(\M_1)'\cap r_\H(\M_2)'$, indeed $z_\H(t)$ belongs to the center of $\ell_\H(\M_1)'\cap r_\H(\M_2)'$ by property $(c)$. 

Clearly, 
\ben\label{z}
z_{\H\otimes\K} = z_\H\otimes z_\K, \quad z_{\H\oplus\H'} = z_\H\oplus z_{\H'}\quad {\rm and} \quad
z_{\bar\H} = \overline{z_\H} \ .
\een 
In particular, if $\H$ is irreducible, then $z(t)\in\mathbb C$, $|z(t)| = 1$. 
In general, since $\H$ has finite index, $\ell_\H(\M_1)'\cap r_\H(\M_2)'$ is finite dimensional, so $\H$ is direct sum of irreducible bimodules, hence $z_\H$ is determined by its value on irreducible bimodules. The generator $\mu_\H$ of $z_\H$, i.e. $z_\H(t) = e^{it\mu_\H}$, is called the {\it chemical potential} of $V_\H$. Obviously, $\mu_{\bar\H} = -\mu_\H$, with the transpose
identification of the center of
$\ell_\H(\M_1)'\cap r_\H(\M_2)'$ with the center of $\ell_{\bar\H}(\M_2)'\cap r_{\bar\H}(\M_1)'$, $T\leftrightarrow \bar T^*$.  
We rephrase these comments in the following proposition. 
\begin{proposition}\label{unique}
In Theorem \ref{th2}, $U_\H$ is characterised by the properties \eqref{modU}, $(a)$, $(b)$, $(c)$ and the symmetry:
\[
\overline{K_\H - \log \Delta_{\H}} = K_{\bar\H} - \log \Delta_{\bar \H} \ .
\]
\end{proposition}
\noindent
Possibly, a non trivial chemical potential compatible with the tensor categorical structure of all finite index bimodules may not exist. It however represents an important quantity within specific tensor categories, see \cite{L01}.

\subsection{Completely positive maps and dilations}
Let $\N$ and $\M$ be von Neumann algebras and $\a : \N \to \M$ a completely positive, normal, unital map. Thus $\a$ is a linear map from $\N$ to $\M$ which is normal (equivalently, continuous in the ultra-weak topology), preserving the unity, such that the natural map $\a\bigotimes{\rm id}_k : \N\bigotimes {\rm Mat}_k(\mathbb C) \to \M\bigotimes {\rm Mat}_k(\mathbb C)$ is positive for every $k\in\mathbb N$. It follows that $||\a||= 1$ and 
$
\a(n^* n) \geq \a(n)^*\a(n)$, $ n\in\N$.
\begin{lemma}\label{pp}
Let $\a : \N\to \M $ be a completely positive, normal, unital map as above and $\F$ a von Neumann algebra. 
The map $\a\bigotimes {\rm id} : \N\odot\F \to \M\odot\F$ extends to a positive, normal map $ \N\bigotimes\F \to \M\bigotimes\F$. 
\end{lemma}
\begin{proof}
Since $\F$ is contained in $B(\K)$ with $\K$ the underlying Hilbert space of $\F$, we may assume that $\F = B(\K)$. We need to show that $\a\bigotimes {\rm id}$ is positive and normal. 
With $e_n \in\F$ a sequence of finite rank projections increasing to 1, clearly $\a\bigotimes {\rm id}$ is positive on
$\cup_n \big(\N\bigotimes e_n\F e_n\big)$ because $\a$ is completely positive. So, since $\cup_n e_n\F e_n$ is ultra-weakly dense in $\F$, it suffices to show that $\a\bigotimes {\rm id}$ is normal. Now, if $\f$ and $\psi$ belongs to the predual of $\N$ and $\F$, the linear functional $\f\bigotimes\psi \cdot \a\bigotimes {\rm id} = \f\cdot\a\bigotimes\psi$ belongs to the predual of $\M\bigotimes\F$. Thus the transpose of $\a\bigotimes {\rm id}$ (as linear bounded operator) maps a total set of normal linear functionals into normal linear functionals, so it is normal. 
\end{proof}
Let $\M$ be a von Neumann algebra and $\f$ a normal faithful state of $\M$. Consider the identity  $\M - \M$ bimodule  $L^2(\M)$ and $\xi_\f\in L^2(\M)_+$ the vector representative of $\f$.
 The bilinear form on $\M$ associated with $\f$ is 
\[
\langle m_1 , m_2\rangle_\f \equiv (\xi_\f, m_1 \xi_\f m_2) = (\xi_\f , m_1 Jm^*_2 J\xi_\f )\ ,
\]
with $J$ the modular conjugation in $L^2(\M)$. 
\begin{proposition}\label{cpbi}
Let $\a:\N\to\M$ be a normal, unital, completely positive map as above and $\f$ a faithful normal state of $\M$.
There exists an $\N - \M$ bimodule $\H_\a$ and a unit vector $\xi\in \H_\a$ such that
\begin{equation}\label{sp}
(\xi, \ell_{\H_\a}(n)r_{\H_\a}(m)\xi) = \langle \a(n), m\rangle_\f \ ,
\end{equation}
and $\xi$ is cyclic ${\H_\a}$ (i.e. for $\ell_{\H_\a}(\N)\vee r_{\H_\a}(\M)$).

Let $\K$ be an $\N-\M$ bimodule with cyclic vector $\eta\in\K$. If
\[
\langle \a(n), m\rangle_\f = (\eta, \ell_\K(n) r_\K(m)\eta)\ ,
\]
there exists an $\N -\M$ bimodule unitary equivalence $U:\K\to\H_\a$ such that $U\eta =\xi$. 
\end{proposition}
\begin{proof}
The map
\[
n, m^o \in \N \times \M^o \mapsto  \langle \a(n), m \rangle_\f 
\]
is bilinear, so it gives a linear map $\tilde\f$ on $\N\odot\M^o$
\begin{equation}\label{tilde}
\tilde\f(n {\textstyle\bigotimes} m^o) = \langle \a(n), m \rangle_\f  \ .
\end{equation}
Now
\[
\tilde\f = \f_\xi\cdot\pi\cdot \a{\textstyle\bigotimes} {\rm id}
\]
on $\N\odot\M$, where $\pi$ is the representation of $\M\odot\M^o$ giving the identity $\M-\M$ bimodule $L^2(\M)$ and $\f_\xi = (\xi,  \cdot\ \xi)$ with $\xi$ the vector representative of $\f$ in $L^2_+(\M)$; thus $\tilde\f$ is positive being the composition of positive maps by Lemma \ref{pp}. 

Let $\H_\a$ be the $\N-\M$ bimodule given by the GNS representation of $\tilde\f$. Clearly eq. \eqref{sp} holds with $\xi$ the GNS vector of $\tilde\f$.

The uniqueness of $\H_\a$ follows by the uniqueness of the GNS representation. 
\end{proof}
Let $\N$, $\M_1$ and $\M_2$ be von Neumann algebras, $\H$ an $\M_1-\M_2$ bimodule and $\r:\N\to\M_1$  a normal homomorphism. 
The $\N-\M_2$ bimodule ${}_\r \H$ is defined as follows: ${}_\r \H = \H$ as a Hilbert spaces, the right actions $r_{\!_\r \H}$ and $r_\H$
of $\M_2$ on $\H$  are the same, while the left action $\ell_{_\r \H}$ of $\N$ on $\H$ is the twist by $\r$ of the left action $\ell_\H$ of $\M_1$, namely $\ell_\r \equiv \ell_{_\r \H} \cdot \r$.

Similarly, if $\theta : \N\to \M_2$ is a normal homomorphism, we can define the $\M_1 -\N$ bimodule $\H_\theta $ by twisting the right $\M$ action on $\H$ by $\theta$, and also ${}_{\r}\H_\theta$ is the $\N-\N$ bimodule where
both left and the right actions on $\H$ are twisted.

An important special case is given by twisting the identity bimodule $L^2(\M)$ of a von Neumann algebra $\M$ by a homomorphis $\r:\N\to\M$. For the
$\N-\M$ bimodule ${}_\r L^2(\M)$, the left action $\ell_\r$ of $\N$ and the right action $r_\r$ of $\M$  on $L^2(\M)$  are given by
\[
\ell_\r(n)r_\r(m) \eta  = \r(n)\eta m \ , \quad \eta\in L^2(\M),\ n\in \N,\ m\in \M \ ,
\]
where the right hand side carries the standard actions of $\M$ on on $L^2(\M)$. 

It can be easily seen that ${}_{\r_1} L^2(\M)$ is unitarily equivalent to $_{\r_2}L^2(\M)$ iff there exists a unitary $u\in\M$ such that $\r_2(n) = u\r_1(n) u^*$. 

\begin{corollary}\label{l2r}
Let $\N$, $\M$ be properly infinite von Neumann algebras, $\H$ an $\N-\M$ bimodule.
Then $\H$ is unitarily equivalent to $_{\r}L^2(\M)$ for some homomorphism $\r:\N\to\M$.

Similarly, $\H$ is unitarily equivalent to $ L_\theta^2(\N)$ for some homomorphism $\th :\M\to \N$.
\end{corollary}
\begin{proof}
Since $r(\M)$ is properly infinite, with properly infinite commutant as $r(\M)'\supset \ell(\N)$, we may choose an identification $\H$ and $L^2(\M)$ as a right $\M$-modules. Then the left action $\ell$ of $\N$ is our homomorphism $\r:\N\to r(\M)' \simeq \M$. 
The second part is similarly proven. 
\end{proof}
With $\r:\N\to\M$ a homomorphism and $_{\r}L^2(\M)$ as above, the conjugate $\overline{_{\r}L^2(\M)}$ is then a $\M-\N$ bimodule that is unitarily equivalent to 
$_{\bar\r}L^2(\N)$ for some homomorphism $\bar\r: \M\to\N$ by Corollary \ref{l2r}. $\bar\r$ is called the {\it conjugate} homomorphism of $\r$. 

\begin{proposition}\label{conjendo}
$_\r L^2(\M) \simeq \,  L_{\bar\r}^2(\N)$. More generally, we have 
\[
  _{\r \cdot\r_1}L_{\r_2}^2(\M) \simeq \,
_{\r_1}L^2_{\bar\r\cdot\r_2}(\N)  \ ,
\]
where $\r_1 :\M_1\to\N$, $\r_2 :\M_2\to\M$ and $\r : \N\to \M$ are homomorphisms.
\end{proposition}
\begin{proof}
If $U : L^2(\N) \to L^2(\M)$ is a Hilbert space unitary operator implementing $\r$, namely $UnU^* = \r(n)$, $n\in\N$, then $\hat U \equiv J_\N U^* J_\M$ implements $\bar\r$ \cite{L90}. Thus, with the natural left and right actions on $L^2(\M)$ and $L^2(\N)$ and $\eta\in L^2(\M)$, we have
\begin{multline*}
\r \r_1 (m_1)\eta \r_2(m_2) = U\r_1 (m_1)U^*J_\M \r_2(m^*_2)J_\M \eta
\sim \r_1 (m_1)U^*J_\M \r_2(m^*_2)J_\M U\eta \\
= \r_1 (m_1)J_\N\hat U \r_2(m^*_2)\hat U^* J_\N \eta'
= \r_1 (m_1)J_\N \bar\r \r_2(m^*_2) J_\N \eta'
= \r_1 (m_1) \eta'  \bar\r \r_2(m_2) \ ,
\end{multline*}
showing the lemma. 
Here $\sim$ denotes the unitary equivalence given by $U$ and $\eta' = U\eta$. 
\end{proof}
We now define the bimodule {\it tensor product} \cite{C, Sau83}. Let $\M_1 , \M_2 , \M_3$ be properly infinite von Neumann algebras, $\H$ a $\M_1 - \M_2$ bimodule, $\K$ a $\M_2 - \M_3$ bimodule. Then $\H\otimes\K$ will be a $\M_1 - \M_3$ bimodule. We consider indeed the {\it relative tensor product}:
we choose a faithful normal positive linear functional $\f$ on $\M_2$ and define the tensor product $\H\otimes_\f \K$ w.r.t. $\f$, yet $\H\otimes_\f \K$ does not depend on $\f$ up to unitary equivalence. So $\H\otimes\K$ will be defined up to unitary equivalence with $\H\otimes_\f \K$ a canonical representative for $\H\otimes\K$. We usually denote $\H\otimes_\f \K$ simply by $\H\otimes\K$.
 
We may identify $\H$ with  $L^2(\M_2)$ as a right $\M_2$-module 
(canonically by the GNS representation of $\f$) and 
$\K$ with $L^2(\M_2)$ as a left $\M_2$-module (again canonically by the GNS representation of $\f$). By definition,  $\H\otimes\K$ is the Hilbert space
$L^2(\M_2)$ where the left  module $\M_1$ action comes from the left action of $\M_1$ on $\H$ and 
the right module $\M_3$ action comes from the right action of $\M_3$ on $\K$.

In other words, we use $\f$ to make the bimodule identification $\H = {}_{\r_1}L^2(\M_2)$ and 
$\K = L^2_{\r_2}(\M_2)$ with homomorphisms $\r_1 :\M_1\to\M_2$ and $\r_2:\M_3\to \M_2$, and then set
\[
{}_{\r_1} L^2(\M_2)\otimes L_{\r_2}^2(\M_2) = {}_{\r_1}L^2_{\r_2}(\M_2) \ .
\]
The uniqueness of $\H\otimes\K$ follows 
by the uniqueness of the standard form $L^2(\M_2), L^2(\M_2)_+$. 
By Proposition \ref{conjendo}, we have 
\[
L^2_{\bar\r_1}(\M_1)\otimes L^2_{\r_2}(\M_2)  \simeq\  L_{\bar\r_1\r_2}^2(\M_1)\ .
\]
This shows, in particular, that the tensor product is associative, up to unitary equivalence. 

Notice that the distributive law holds: if $e_k$ is a partition of the unity by projections in $\ell_\H(\M_1)'\cap r_\H(\M_2)'$, then
\[
\H\otimes_\f\K = \bigoplus_k (\H_k \otimes_\f \K)
\]
with $\H_k \equiv e_k \H$ (we have $(e_k \H)\otimes\K = (e_k\otimes 1)(\H\otimes\K)$, see below). 

We now define the {\it tensor product of intertwiners}. Let $\H$ and $\H_1$ be $\N-\M$ bimodules with actions $\ell, r$ and $\ell_1, r_1$. Then $T\in {\rm Hom(}\H,\H_1)$ if $T\in B(\H,\H_1)$ and $T$ intertwines $\ell$ with $\ell_1$ and $r$ with $r_1$. Thus, with $\H = {}_\r L^2(\M)$, $\H_1 = {}_{\r_1}L^2(\M)$ and $\r , \r_1 : \N\to \M$ homomorphisms, we have
\[
T\in {\rm Hom}\big({}_\r L^2(\M), {}_{\r_1}L^2(\M)\big) \Longleftrightarrow T\in B\big(L^2(\M)\big), \ T\in\M \ \& \ T\r(n) = {\r_1}(n)T\ , \quad n\in\N \ ,
\] 
where $T\in\M$ means that $T$ belongs to $\M$ acting on the left on $L^2(\M)$. 

Similarly, with $\K = L_{\th}^2(\M)$,
$\K_1 = L_{\th_1}^2(\M)$ $\M-\L$ bimodules and $\th, \th_1: \L \to \M$ homomorphisms of a von Neumann algebra $\L$ into $\M$, we have
\[
S\in {\rm Hom}\big(L_{\th}^2(\M), L_{\th_1}^2(\M)\big) \Longleftrightarrow S\in B\big(L^2(\M)\big), \  S\in\M' \ \& \ S\th(z) = {\th_1}(z)S\ , \ z\in\L \ ,
\]
with $\M$ as above (i.e. $S$ belongs to $\M$ acting on the left on $L^2(\M)$).

Then $T\otimes S\in {\rm Hom}\big(\H\otimes\K, \H_1\otimes\K_1\big)$ is given by 
\[
(T\otimes S)\xi = T\xi S\ , \quad \xi\in L^2(\M) \ ,
\]
(left and right action on $L^2(\M)$), with the Hilbert space bimodule identifications  $\H\otimes\K = {}_{\r}L^2(\M)_\th$ and 
 $\H_1\otimes\K _1= {}_{\r_1}L^2(\M)_{\th_1}$ (thus $\H\otimes\K = \H_1\otimes\K _1= L^2(\M)$ as Hilbert spaces). 

The following proposition shows that the tensor product of intertwiners is equivalent to the DHR notion \cite{DHR}, see \cite{L90}.
\begin{proposition}\label{intert}
Let $\r, \r' : \M_1 \to\M_2$, $\th, \th' : \M_2 \to \M_3$ be normal homomorphisms and
$T\in\ {\rm Hom}\big({}_\r L^2(\M_2), {}_{\r'} L^2(\M_2)\big)$,
$S\in\ {\rm Hom}\big({}_\th L^2(\M_3), {}_{\th'} L^2(\M_3)\big)$ intertwiners, so $T\otimes S \in 
{\rm Hom}\big({}_{\th\r} L^2(\M_3), {}_{\th'\r'} L^2(\M_3)\big)$ and $T\in\M_2$, $S\in\M_3$.
As elements of $\M_3$, we then have
$T{\otimes} S  = S\th(T) = \th'(T)S$.
\end{proposition}
\begin{proof}
 Let $V: L^2(\M_2) \to L^2(\M_3)$ be a unitary implementing $\th$. 
Then $V^*$ implements a unitary equivalence ${}_\th L^2(\M_3) \to  L_{\bar\th}^2(\M_2)$ (see Prop. \ref{conjendo}). By definition, we thus have 
$
T\otimes S = V(T V^*SV)V^* = (VT V^*)S = \th(T)S
$.
\end{proof}
We go back to Theorem \ref{th1} for a moment. 
Since $\Delta^{it} \equiv \Delta^{it}_\H(\f_1|\f_2) \otimes \Delta^{it}_{\K}(\f_2|\f_3)$ is a tensor product of intertwiners, 
it follows by property \eqref{modH} that
\[
\Delta^{it}\ell_{\H\otimes\K}(m_1)\Delta^{-it} = \ell_{\H\otimes\K}\big(\s_t^{\f_1}(m_1)\big)\ , \qquad \Delta^{it}r_{\H\otimes\K}(m_3)\Delta^{-it} = r_{\H\otimes\K}\big(\s_t^{\f_3}(m_3)\big)\ .
\]
This is a main implication of $(a)$ in Theorem \ref{th1}. 
\medskip

\noindent
The above definition of tensor product was given for bimodules over properly infinite von Neumann algebras. The general case may be defined, up unitary equivalence, by considering bimodule ampliations. Namely, if $\H$ a is $\M_1 - \M_2$ bimodule and $\K$ a $\M_2 - \M_3$ bimodule,
we consider a type $I_\infty$ factor $\F$ and the external tensor product bimodules $\H\bigotimes L^2(\F)$ and $\K\bigotimes L^2(\F)$, which  are  a $\M_1\bigotimes \F - \M_2\bigotimes \F $ bimodule and a $\M_2\bigotimes \F  - \M_3\bigotimes \F $ bimodule.  Then, by restricting the left and right actions, $(\H\bigotimes L^2(\F)) \otimes (\K\bigotimes L^2(\F))$ is a $\M_1 - \M_3$ bimodule and
$\H\otimes\K $ is defined to be its reduced bimodule by $1\bigotimes e$, with $e$ a minimal projection of $\F$. 

We now give a generalisation of Stinespring dilation theorem, that allows to dilate completely positive, normal maps between von Neumann algebras to normal homomorphisms, and show the uniqueness of the minimal dilation, cf. \cite{C}, see also \cite{BMSS}. 

Let $\a: \N\to \M$ be a normal, completely positive unital map between the von Neumann algebras $\N$, $\M$. A pair $(\r , v)$ with $\r: \N\to \M$ a homomorphism and $v\in \M$ an isometry such that
\[
\a(n) = v^*\r(n) v\ , \quad n\in \N\ .
\]
will be called a {\it dilation pair} for $\a$, and $\r$ a {\it dilation homomorphism}.

With $(\r, v)$ a dilation pair, the subspace  $\r(\N)v\H$ of the underlying Hilbert space $\H$ is clearly both $\r(\N)$-invariant and $\M'$-invariant, thus the projection $e$ onto $\overline{\r(\N)v\H}$ (i.e. the left support of 
$\r(\N)v\H$\,) belongs to $\r(\N)'\cap \M$. We shall say that $(\r, v)$ is {\it minimal} if $e=1$. 
\begin{theorem}\label{dilation}
Let $\a: \N\to \M$ be a normal, completely positive unital map between the von Neumann algebras $\N$, $\M$, with $\N$, $\M$ properly infinite. 
Then there exists a minimal dilation pair $(\r, v)$ for $\a$. 

If $(\r_1, v_1)$ is another minimal dilation pair for $\a$, there exists a unique unitary $u\in \M$ such that 
\begin{equation}\label{unid}
u\r(n) = \r_1(n)u\, ,\quad v_1 = uv\, ,
\end{equation}
for all $n\in\N$. 

If  $(\r_1, v_1)$ is not minimal, then $\r$ is a subsector of $\r_1$ ($_\r L^2(\M)$ is contained in $_{\r_1} L^2(\M)$). 
\end{theorem}
\begin{proof}
We may assume that $\M$ acts standardly on $L^2(\M)$ with cyclic and separating unit vector $\xi\in L^2(\M)$. 
Let $\f$ be the binormal state on $\N\bigotimes \M^o$ determined by
\[
\f(n\bigotimes m^o) = (\xi, \a(n)\xi m)
\]
Let $\H_\a$ the $\N-\M$ bimodule associated with $\a$ by $\f$, with cyclic vector $\eta\in\H_\a$, given by Proposition \ref{cpbi}. By Proposition \ref{l2r} we may identify $\H_\a$ with $_\r L^2(\M)$ with $\r:\N\to\M$ a homomorphism. 
Then eq. \eqref{sp} reads as an equation in $L^2(\M)$:
\[
(\xi, \a(n)\xi m) = (\eta ,\r(n)\eta m)
\]
for some vector $\eta\in L^2(\M)$ cyclic for $\r(\N)\vee\M'$.  

Since $(\xi,\xi m) = (\eta ,\eta m)$,
the map $v: \xi m \mapsto \eta m$, $m\in\M$, is isometric and its closure is an isometry
$v: L^2(\M)\to L^2(\M)$ with final projection the orthogonal projection $p$ onto $\overline{\eta \M}$.
As $v$ commutes with the right action of $\M$ on $L^2(\M)$, we have $v\in\M$ and  
\[
(\xi, v^*\r(n)v\xi m) = (\eta, \r(n)\eta m) = (\xi, \a(n)\xi m),
\]
thus $\a = v^*\r(\cdot)v$ by Proposition \ref{cpbi}. 

Now, the left support of $\r(\N)v$ is the left support of $\r(\N)p$, namely the projection onto the closure of $\r(\N)\eta \M$, which is equal to 1 by the cyclicity of $\eta$; thus the pair $(\r, v)$ is minimal. 

Now we show the uniqueness of the minimal pair $(\r, v)$. So let $(\r_1, v_1)$ be another minimal dilation pair for $\a$. 
We have
\[
(\xi , v^*\r(n)v\xi m) = (\xi , \a(n)\xi m) = (\xi , v^*_1\r_1(n)v_1\xi m),\quad n\in \N, \ m\in \M,
\]
thus, with $\xi' = v\xi$ and $\xi'_1 = v_1\xi$ we have
\[
(\xi' , \r(n)\xi' m) = (\xi_1', \r_1(n)\xi_1' m),\quad n\in \N, \ m\in \M \ .
\]
By the uniqueness of the GNS representation, there exists a partial isometry $u: L^2(\M)\to L^2(\M)$ such that 
\ben\label{uv}
u\r(n)\xi' m = \r_1(n)\xi_1' m \ .
\een
with $u^* u$  the projection onto $\overline{ \r(\N)\xi' \M}$ and 
$uu^* $ the projection $q$ onto $\overline{ \r_1(\N)\xi_1' \M}$. Since $(\r,v)$ is  minimal, $u^*u =1$, namely $u$ is an isometry in $\M$ such that
\[
u\r(n) = \r_1(n)u
\]
and indeed $u$ is unitary since  $(\r_1,v_1)$ is  minimal too.

Now, by eq. \eqref{uv}, we have
\[
uv\xi m = v_1\xi m
\]
thus $uv = v_1$ since $\xi$ is cyclic for the left action of $\M$. 

We show now that $u$ is unique. If $u\in\M$ is a unitary with properties $\eqref{unid}$, we have
\[
u\r(n)v = \r_1(n)u v = \r(n) v_1 \ ,
\] 
thus $u$ is fixed by the minimality assumption. 

Finally, if $\r_1$ were not minimal, the $u$ in \eqref{uv} would be an isometry, namely $\r$ would be a subsector of $\r_1$. 
\end{proof}
\begin{corollary}\label{uniqdil}
Let $\gA$ be a unital C$^*$-algebra and $\Phi: \gA\to \M$  a completely positive, unital map from $\gA$ into a properly infinite von Neumann algebras $\M$.

There exist an isometry $v\in \M$ and a representation $\r$ of $\gA$ on $\H$ with $\r(\gA)\subset\M$ such that 
\[
\Phi(x) = v^*\r(x) v\ , \quad x\in \gA\ .
\]
\end{corollary}
\begin{proof}
Let $\psi$ be a faithful normal state of $\M$ and $\f \equiv \psi\cdot\Phi$ its pullback to a state of $\gA$. 
Then $\Phi$ factors through the GNS representation of $\gA$ given by $\f$:
\[
\xymatrix{
\gA  \ar[d]_{\pi_\f} \ar[rr]^\Phi & & \M  \\
\N  \ar[urr]_{\Phi_0} & }
\]
with $\N\equiv \pi_\f(\gA)''$ and $\Phi_0 : \N\to \M$ a completely positive map. Indeed if $a\in\gA$ we have
\[
\pi_\f(a) = 0 \implies \f(a^* a) = 0 \implies \psi\cdot\Phi(a^* a ) = 0 \implies \Phi(a^* a ) = 0 \implies \Phi(a ) = 0
\]
since $\Phi(a)^*\Phi(a) \leq \Phi(a^* a)$. As $\psi\cdot \Phi_0$ is normal on $\N$, it follows easily that $\Phi_0$ is normal.

We now apply Theorem \ref{dilation} to $\Phi_0$. If $\N$ is properly infinite we get immediately our statement. In general, we may consider $\N\bigotimes \F$, with $\F$ a type $I_\infty$ factor, and a faithful normal conditional expectation $\varepsilon : \N\bigotimes \F \to\N$, apply Theorem \ref{dilation} to $\Phi_0\cdot\e$ and read the formula for $\Phi_0 = \Phi_0\cdot\e |_\N$. 
\end{proof}
The following corollary extends to the infinite-dimensional case the known construction of Kraus operators. 
\begin{corollary}
Let $\a:\F\to\F$ be a completely positive, normal, unital map with
$\F$ a type $I_\infty$ factor. There exists a sequence of elements $T_i\in\F$ with $\sum_i T_i T_i^* = 1$ such that 
\[
\a(x) = \sum_i T_i x T_i^*\, ,\quad x\in\F \ .
\]
\end{corollary}
\begin{proof}
Write $\a= v^*\r(\cdot)v$ by Theorem \ref{dilation}, with $(\r,v)$ a minimal dilation pair for $\a$. As shown in  \cite{L87}, every endomorphism of a type $I$ factor is inner, namely there exists a sequence of isometries $v_i\in\F$ 
with $\sum_i v_i v_i^* = 1$ such that 
\[
\r(x) = \sum_i v_i x v^*_i \ ,\quad x\in\F \ .
\]
Thus
\[
\a(x) = v\r(x)v^* = \sum_i v^* v_i x v^*_i v = \sum_i T_i x T_i^*\ , \quad x\in \F\ ,
\]
with $T_i = v^* v_i$.  
\end{proof}
\begin{corollary}\label{dilcor}
Let $\a:\N\to\M$ be a normal, completely positive, unital map, with $\M$ properly infinite. With $\f$ a normal faithful state of $\M$, the bimodule $\H_\a$ associated with $\a$ is unitarily equivalent to $_\r L^2(\M)$, with $\r$ the minimal dilation of $\a$. In particular, $\H_\a$ does not depend on $\f$ up to unitary equivalence. 
\end{corollary}
\begin{proof}
Immediate by Theorem \ref{dilation} and its proof.
\end{proof}
Let $\a :\M_1\to\M_2$ and $\b :\M_2\to\M_3$ be completely positive, normal, unital maps of properly infinite von Neumann algebras and $\r$, $\th$ the corresponding minimal dilation homomorphisms. Then $\th\r$ is a dilation homomorphism of $\b \a$: 
write $\a = v^*\r(\cdot)v$ and $\b = w^*\th(\cdot)w$; then $\b \a =w^*\th( v^*)\th\r(\cdot)\th(v)w$. In general, $\th\r$ could fail to be minimal, yet the minimal dilation of $\b\a$ is a contained in $\th\r$ by Theorem \ref{dilation}. 

Before concluding this section, we comment on the transpose of a completely positive map, see \cite{OP}. 

\begin{proposition}\label{transpose}
Let $\N, \M$ be von Neumann algebras and $\a:\N\to\M$ a normal, completely positive, unital, faithful map. 
Fix a faithful normal state $\f$ of $\M$ and consider the state $\psi \equiv \f\cdot \a$ on $\N$. 

There exists a unique normal, completely positive, unital map $\a':\M\to\N$ such that
\begin{equation}\label{a'}
\langle \a(n), m\rangle_\f = \langle \a'(m), n\rangle_\psi \ . 
\end{equation}
$\a'$ is called the transpose of $\a$. 
\end{proposition}
\begin{proof}
Denote by $\iota_\N : \N\to \N_*$ the embedding given by $\psi$, namely $\iota_\N(n) = (J_\N n\xi_\psi , \,\cdot\,  \xi_\psi)$ on $L^2(\N)$. Then $\iota_\N$ is completely positive (with the natural matrix order structure on preduals) and
the range of $\iota_\N$ consists of the elements of $\N_*$ that are dominated by $\psi$. 
Similarly for the embedding $\iota_\M : \M\to \M_*$ given by $\f$.  Then $
\a'_\f = \iota^{-1}_\N\cdot\a^{\rm t}\cdot\iota_\M $, with $\a^{\rm t}:\M_*\to\N_*$ the pre-transposed map of $\a$ as linear operator.  
\end{proof}
Note that $\a'$ depends on the state $\f$, while the conjugate of a homomorphism is unique up to equivalence.  Yet, we have the following. 
\begin{proposition}\label{diltrans}
With $\a$ and $\a'$ as above, then $\H_{\a'}$ is unitarily equivalent to $\bar\H_\a$. 

In particular, if $\N,\M$ are properly infinite,  $\a$ and  $\a'$ have conjugate minimal dilation homomorphisms. 
\end{proposition}
\begin{proof}
By eq. \eqref{a'} the natural state $\tilde\f$ and $\tilde\psi$ on $\N\odot\M^o$ and $\M\odot\N^o$ as in eq. \eqref{tilde} are related by $\tilde\f(n\bigotimes m^o) = \tilde\psi(m\bigotimes n^o)$. The statement follows by eq. \eqref{conj}. 
\end{proof}
\section{Quantum channels}
\label{QC}
Let $\N, \M$ be von Neumann algebras with finite dimensional centers $Z(\N), Z(\M)$. 

The {\it index of a completely positive, normal, unital map $\a:\N\to\M$} is now defined by
\[
\Ind(\a) \equiv \Ind(\H_\a)\ ,
\]
with $\H_\a$ the $\N-\M$ bimodule associated with $\a$. 

Therefore, if $\N$ and $\M$ are properly infinite,
\[
\Ind(\a) = \Ind(\r) \ ,
\]
where $\r:\N\to\M$ is the minimal dilation homomorphism of $\a$ and $\Ind(\r) = [\M : \r(\N)]$. 

By a {\it quantum channel} $\a:\N\to\M$ we shall mean a completely positive, normal, unital map $\a$ with finite index. 
Let $\a:\N\to\M$ be a quantum channel and $\fin$ a faithful normal state of $\M$. We choose a faithful normal state $\fin$ of $\M$, the {\it input} state for $\a$. The {\it output} state on $\N$ is then defined by 
\[
\fout = \fin\cdot \a\ .
\]
With $\H_\a$ the $\N-\M$ bimodule associated with $\a$ by $\fin$, with cyclic vector $\xi\in\H_\a$ as in Proposition \ref{cpbi},  we have
\[
(\xi, r(m)\xi ) = \fin(m)\ ,\quad (\xi, \ell(n)\xi) = \fout(n) \ ,
\]
namely
\[
\fin = \f_\xi\cdot r\ , \quad \fout = \f_\xi\cdot\ell \ ,
\]
where $\f_\xi$ is the vector state associated with $\xi$ and $\ell,r$ the left and right actions of $\N,\M$ on $\H_\a$. 

Clearly the transpose map $\a'$ w.r.t. $\fin$ interchanges $\fin$ with $\fout$:
\[
\fin = \fout\cdot \a'\ .
\] 
Let $\e : r(\M)'\to \ell(\N)$ be the minimal expectation. Then 
\[
\Phi = \ell^{-1}\cdot \e : r(\M)'\to \N
\]
is a quantum channel from $r(\M)'\simeq \M$ to $\N$ called the {\it left inverse} of $\a$. Notice that $\Phi\cdot\ell$ is the identity on $\N$. 

We define the {\it modular operator} $\Delta_{\a,\fin}$ of $\a$ with respect to the initial state $\fin$ as the modular 
operator of the associated $\N-\M$ bimodule $\H_\a$ with respect to the initial and the final state:
\begin{equation}\label{delta}
\Delta_{\a,\fin} = \Delta_{\H_\a}(\fout |\fin) 
= d\fout\cdot\ell^{-1}\cdot\e \big/ d\fin\cdot r^{-1}\ ,
\end{equation}
thus $
\Delta_{\a,\fin} = d( \f_\xi \cdot\e) \big/ d(\f_\xi \big\vert_{r(\M)})$. Then $\Delta_{\a,\fin}$ is a positive, non-singular selfadjoint operator on $\H_\a$ and we have
\begin{equation}\label{impl}
\Delta_{\a,\fin}^{it}\ell(n)\Delta_{\a,\fin}^{-it} = \ell\big(\s^{\rm out}_t(n)\big)\ , \quad \Delta_{\a,\fin}^{it}r(m)\Delta_{\a,\fin}^{-it} 
= r\big(\s_{t}^{\rm in}(m)\big) \ ,
\end{equation}
where $\s^{\rm in/out}$ is the modular group of $\M$/$\N$ w.r.t. $\f^{\rm in/out}$. 

The {\it entropy of a quantum channel} $\a$ in the initial state $\fin$ is now defined as 
\ben\label{entroa}
S_{\a , \fin} \equiv -(\xi ,\log\Delta_{\a,\fin} \xi)\ .
\een
Assume that $\N$, $\M$ are properly infinite. Let's write the above formulas in terms of the minimal dilation pair $(\r, v)$ for $\a$.

By Corollory \ref{dilcor}  $\H_\a = {}_\r L^2(\M)$. So $\ell = \r$, the $\r$-twisted left action of $\N$ on $L^2(\M)$, and $r$ is the  standard right action of $\M$ on $L^2(\M)$.   We choose a faithful normal state $\fin = \o$ on $\M$; then $\fout = \o\cdot\a = \o(v^*\r(\cdot)v)$ 
is the output state on $\N$. 

With $\e : \M \to \r(\N)$ the minimal expectation, $\Phi = \r^{-1}\cdot\e : \M\to\N$ is the left inverse of $\r$. We have 
$\fout\cdot \Phi = \o\cdot\a \cdot\Phi = \o(v^*\e(\cdot)v)$, so
\[
\Delta_{\a,\fin} = \Delta_{{}_\r L^2(\M)}(\fout |\fin) =
d\fout \Phi\big/d\o = d(\o(v^*\e(\cdot)v))\big/ d\o = \Delta_{\hat\xi,\xi} \ ,
\]
where $\xi,\hat\xi$ are the vector representatives in $L^2(\M)_+$ of the states $\o$ and $\o(v^*\e(\cdot)v)$ on $\M$ and
$\Delta_{\hat\xi,\xi}$ is the relative modular operator associated with $\hat\xi, \xi$. Thus 
\[
S_{\a , \fin} = -(\xi , \log\Delta_{\hat\xi,\xi}\xi)
\] 
is indeed Araki's relative entropy with respect to the states $\o$ and $\o(v^*\e(\cdot)v)$ of $\M$.  In particular, we have obatained the {\it positivity of the entropy}:
\begin{proposition}
$S_{\a , \fin} \geq 0$. 
\end{proposition}
\begin{proof}
It remains to consider the case $\N$ or $\M$ are not properly infinite. Yet, with $\F$ a type $I_\infty$ factor, we have $S_{\a , \fin} = S_{\tilde\a , \tilde\f_{\rm in}} \geq 0$, with $\tilde\a = \a\bigotimes {\rm id}$ the ampliation by tensoring with $\F$ and $\tilde\f_{\rm in} = \fin \bigotimes \f'$ and $\f'$ any normal faithful state of $\F$. 
\end{proof}
\subsection{Lower bound for the free energy}
The modular group $t\mapsto \s_t^\o$ of a von Neumann algebra $\M$ is an intrinsic dynamics associated with a faithful normal state $\o$ of $\M$. However, the  evolution parameter $t$ is canonical up to a scaling in the sense it is uniquely determined when we fix the inverse temperature $\b = 1/kT$, namely $t\mapsto \tau_t = \s_{-{\b^{-1}} t}^\o$ is the unique one parameter automorphism group of $\M$ that satisfies the $\b$-KMS condition with respect to $\o$ as stated in \eqref{KMS}. 

The modular unitary group $t\mapsto \Delta_\o^{it}$ is a canonical implementation of $\s^\o$ on $L^2(\M)$ (with respect to the vector representative of $\o$ in $L^2(\M)_+$):
\[
\Delta_\o^{it} m \Delta_\o^{-it} =\s^\o_t(m)\ ,\quad m\in\M \ .
\]
Analogously, the {\it modular group of a quantum channel} $\a : \N\to\M$ is an {\it intrinsic dynamics} associated with $\a$ and an initial (faithful, normal) state $\fin$ of $\M$; and the choice of the inverse temperature $\b$ uniquely fixes the scaling parameter. 

The modular unitary group $\Delta_{\a,\fin}^{it} \equiv \Delta^{it}_{\H_\a}(\fout |\fin)$ of $\a$ is a canonical implementation of  the modular automorphism group $\s_t^{\rm out}\odot \s_{t}^{\rm in}$ of $\N\odot\M$ on the $\N - \M$ bimodule $\H_\a$, 
namely 
\begin{align}\label{impl2}
\Delta_{\a,\fin}^{it} \ell(n)\Delta_{\a,\fin}^{-it}  &= \ell\big(\s^{\rm out}_t(n)\big)\ , \\ 
\Delta_{\a,\fin}^{it} r(m)\Delta_{\a,\fin}^{-it}  &= r\big(\s_{t}^{\rm in}(m)\big) \ ,
\end{align}
where $\fout = \fin\cdot\a$. Here we are always assuming that the centers of $\N$ and $\M$ are finite dimensional.

As seen, there exists another canonical implementation of $\s_t^{\rm out}\odot \s_{t}^{\rm in}$ on $\H_\a$, that 
naturally behaves under the tensor categorical structure provided by finite index bimodules. By the discussion in Section \ref{CB},
we thus define
\[
U^\a_t(\fout|\fin) \equiv U^{\H_\a}_t(\fout|\fin) \ .
\] 
$U^\a$ is symmetric with respect to the transpose map by Proposition \ref{diltrans}
\[ 
U_t^{\a'}(\fin |\fout) =  \overline{U^\a_{t}(\fout |\fin)}\ ,
\]
and is given by
\[
U_t^\a(\fout |\fin) = \Delta_{\a,\fin}^{it}D(\a)^{it} \ ,
\]
where $D(\a)$ is a central operator associated with the matrix dimension of $\H_\a$. 
We call  $U^\a$ the {\it physical unitary evolution} associated with $\a$. 

More generally, when we consider the $\b$-rescaling of the modular parameter, we get the {\it physical unitary evolution} $t\mapsto U_\a(-\b^{-1} t)$ at inverse temperature $\b$. 
Its self-adjoint generator on $\H_\a$ is the {\it physical Hamiltonian} $H_{\a , \fin}$ associated with the quantum channel $\a$ and the state $\fin$ at inverse temperature $\b$.  

By the above discussion, we have
\begin{equation}\label{Ha}
\b H_{\a,\fin} = -\log \Delta_{\a,\fin} - \log D(\a) \ .
\end{equation}
The {\it mean energy} $E_{\a,\fin}$ of $\a$ is now defined by
\[
E_{\a , \fin} \equiv (\xi, H_{\a , \fin} \xi) \ ,
\]
with $\xi$ the cyclic vector in $\H_\a$ associated with $\fin$.

Then the {\it (incremental) free energy} $F_{\a, \fin}$ is defined by the thermodynamical relation
\begin{equation}\label{thermorel}
F_{\a , \fin} = E_{\a , \fin} - \b^{-1}S_{\a , \fin} \ ,
\end{equation}
where the {\it entropy} $S_{\a , \fin}$ of $\a$ relative to the initial state $\fin$  is defined in \eqref{entroa} as a relative entropy.

Note that
\[
F_{\a , \fin} = F_{\a' , \fout}
\]
and this relation, together with natural relations, fixes the choice of the physical Hamiltonian by Proposition \ref{unique}.
\begin{theorem}\label{freeen}
We have:
\begin{equation}\label{free}
F_{\a , \fin} = -\b^{-1}(\xi, \log D(\a)\,\xi) \ .
\end{equation}
\end{theorem}
\begin{proof}
By evaluating both sides of equation \eqref{Ha} on the vector state associated with $\xi$ we have
\[
\b (\xi, H_\a \xi) = -(\xi, \log \Delta_{\a,\fin} \xi) - (\xi, \log D(\a)\xi)\ ,
\]
namely
\[
E_{\a , \fin} - \b^{-1} S_{\a , \fin} =  - \b^{-1}(\xi, \log D(\a)\xi) \ ,
\]
thus eq. \eqref{free} follows by the thermodynamical relation \eqref{thermorel}. 
\end{proof}
\begin{corollary}
If $\N$ and $\M$ are factors,
the free energy is expressed by the {\it relative partition formula}:
\ben\label{free2}
F_{\a , \fin} \equiv -\b^{-1}\log(\xi, e^{-\b H_{\a , \fin}}\xi) \ .
\een
\end{corollary}
\begin{proof}
We may identify $\H_\a$ with $_\r L^2(\M)$, where $\r$ is the minimal dilation homomorphism of $\a$. Then $\Delta_{\a , \fin} = \Delta_{\xi, \eta}$, where $\Delta_{\xi, \eta}$ is the relative modular operator w.r.t. the vectors  $\xi$ and $\eta$   in $L^2(\M)_+$
giving the states $\f$ and $\f\cdot \Phi$, with $\Phi$ the  left inverse of  $\r$. 

As $\a$ is factorial, $D(\a) = d_{\H_\a}$ is a scalar, so we have
\begin{multline}
(\xi, e^{-\b H_{\a , \fin}}\xi) = D(\a)(\xi, \Delta_{\a , \fin}\xi) =
 D(\a)(\Delta^{1/2}_{\xi, \eta}\xi, \Delta^{1/2}_{\xi, \eta}\xi) \\
= D(\a)(J\Delta^{1/2}_{\xi, \eta}\xi, J\Delta^{1/2}_{\xi, \eta}\xi) = 
D(\a)(S_{\xi, \eta}\xi, S_{\xi, \eta}\xi) = D(\a)(\eta , \eta) = D(\a)
\end{multline}
with $J$ the modular conjugation on $L^2(\M)$.
The corollary then follows by Theorem \ref{freeen}. 
\end{proof}
Since $D(\a)\geq 1$, as a consequence of Theorem \ref{freeen} we have the {\it negativity of the incremental free energy}
\[
- F_{\a , \fin} \geq 0\ .
\]
Notice that, if $D(\a)$ is a scalar, we have
\[
F_{\a , \fin} = - \b^{-1}\log d(\a)\ ,
\]
which is independent of $\fin$; this holds, in particular, in the factorial case.

In general, we define the {\it free energy of a quantum channel} $\a$ as
\[
F_\a \equiv \inf F_{\a , \fin} \ ,
\]
infimum over all (faithful, normal) initial states $\fin$. 
\footnote{By eq. \eqref{Ko}, choosing the right modular operator $\Delta_{\a,\fin}'$, rather than the left modular operator $\Delta_{\a,\fin}$, would give $\b H_{\a,\fin} = -\log \Delta'_{\a,\fin} + \log D(\a)$ (eq.  \eqref{Ha}), so this would amount in defining the free energy as $-F_\a$, rather than $F_\a$. With this convention, we would thus have the {\it positivity of the free energy}. Our choice of $\Delta_{\a,\fin} $ here is consistent with the one in \cite{L97}. 

If one considered a Hamiltonian with non trivial chemical potential (see comments at the end of Section \ref{CB}), by Proposition \ref{diltrans} and last eq. \eqref{z} one should define the free energy as $\frac12 (F_{\a'} + F_\a)$, see \cite{L01}. }

By Jones' theorem \cite{Jo}, the dimension (i.e. the square root of the index) $d$ of a subfactor is quantised: if $d<\infty$ then
\[
d = 2\cos(\pi/n) \ , \quad n=3,4,\dots\  {\rm or}\ d\geq 2 \ ;
\] 
in particular
\[
d \neq 1 \implies d\geq \sqrt{2} \ ,
\]
so we have the following corollary.
\begin{corollary}
If $F_{\a}\neq 0$, then
\ben\label{ineq}
-F_\a \geq{\textstyle \frac12} kT\log 2\ .
\een
\end{corollary}
\begin{proof}
If both $\N$ and $\M$ are factors, $D(\a) = d_{\H_\a}$ is a scalar, so $F_\a =F_{\a,\fin} = -\b^{-1} \log d_{\H_\a}$
by Theorem \ref{freeen} and \eqref{ineq} is immediate by Jones' theorem.

In general, we choose a faithful normal state $\f = \fin$ of $\M$ and let $\H_\a$ be the $\N=\M$ bimodule associated to $\a$ by $\f$, with cyclic vector $\xi$ (Prop. \ref{cpbi}), namely
\[
( \xi  ,  \ell_{\H_\a}(n) r_{\H_{\a}}(m)\xi)  = (\xi_\f , \a(n)\xi_\f m) \ ,
\]
where $\xi_\f$ is the vector representative of $\f$ in $L^2(\M)_+$. 

Let then $q$ be  projection of  $Z(\M)$ and $\a_q :\N\to \M_q$ given by $\a_q(n) = \a(n)q$.  We want to show that the $\N-\M_q$ bimodule $\H_{\a_q}$ associated with $\a_q$ by $\f_q \equiv \f|_{\M_q}$ is $r_{\a}(q) \H_\a$ with the reduced left and right actions $\ell_{\H_{\a_p}}(n) = \ell_{\H_\a}(n)|_{r_{\H_\a}(q) \H_\a}$ and $r_{\H_{\a_p}}(mq) =
r_{\H_\a}(mq)|_{r_{\H_\a}(q) \H_\a}$, and cyclic vector $\xi_q = r_{\H_{\a_p}}(q) \xi$.
We have indeed
 \[
( \xi_q  ,  \ell_{\H_\a}(n) r_{\H_{\a_q}}(mq)\xi_q) = ( \xi  ,  \ell_{\H_\a}(n) r_{\H_{\a_q}}(mq)\xi) = (\xi_\f , \a(n)\xi_\f qm) = (\xi_\f , \a_q(n)\xi_\f^q m) \ ,
 \]
 where $\xi^q_\f \equiv \xi_\f q$ is the cyclic vector in $L^2(\M q)_+ = qL^2(\M)_+ $ associated with $\f_q$. 
 
With $\{q_j\}$ the minimal central projections of $\M$, we thus have
 \begin{multline}
 F_{\a , \f} = -\b^{-1}(\xi, \log D(\a)\,\xi) =  \sum_j -\b^{-1}(\xi_{q_j}, \log D(\a)\,\xi_{q_j}) \\
 =   \sum_j -\b^{-1}(\xi_{q_j}, \log D(\a_{q_j})\,\xi_{q_j}) = \f(q_j)\sum_j F_{\a_{q_j} , \psi_{q_j} }
 \end{multline}
 with $\psi_{q_j} \equiv \f_{q_j}/|| \f_{q_j} ||$ \!. 

  If $ F_{\a , \f}\neq 0$, there exists a $j'$ such that $F_{\a_{q }, \psi_{q }}\neq 0$ with $q\equiv q_{j'}$ and, in particular,
  \ben\label{fineq}
 F_{\a , \f} \geq \f(q) 
F_{\a_{q} , \psi_{q}} \ .
  \een
Assume now $\N$ to be a factor.
Then,  $F_{\a_{q}, \psi_{q}}\geq  {\textstyle \frac12} kT\log 2$ since $\a_{q}$ is factorial. By choosing $\f$ with $\f(q)> 1 -\epsilon$, we infer from \eqref{fineq} that $F_{\a} \geq {\textstyle  \frac12} kT\log 2$. Thus the corollary holds if $\N$ is a factor. 

On the other hand, the inequality \eqref{ineq} also holds if $\M$ is a factor and $\N$ arbitrary: it is sufficient to consider the transpose channel $\a'$.

The general case now follows again by the inequality \eqref{fineq}, by taking a normal faithful state $\f$ on $\M$ such that 
$F_{\a_{q} , \psi_{q}} \geq  (1- \epsilon) F_{\a_{q}} $ and $\f(q) \geq (1 - \epsilon)$. Such $\f$ can be constructed as follows: first choose a state $\psi$ on $\M_{q} $ such that 
$F_{\a_{q} , \psi_{q}} >  (1- \epsilon) F_{\a_{q}} $ and a normal faithful state $\o$ on $\M_{1-q}$; then define the state $\f$ by  $\f(m) = (1-\epsilon)\psi(qmq) + \l\o((1-q)m(1-q))$, with $   \l  = -(1 - \epsilon)/(1-\o(q)) $. 
\end{proof}
The above corollary gives a {\it general lower bound} for the positive free energy $-F_\a$. It is half of the Landauer bound derived in \cite{La}. 
\medskip

If $\a$ is a quantum channel between finite dimensional quantum systems $\N,\M$ or, more generally, if $\N,\M$ are discrete von Neumann algebras, as the dimension an inclusion of type $I$ factors is a positive integer, the above proof gives
\begin{equation}\label{Lb}
-F_\a \geq kT\log 2\ ,
\end{equation}
which is indeed the lower bound derived by Landauer in a finite dimensional context.
\medskip

Quantum Field Theory provides examples for the above setting.  A DHR charge $\a$ is an endomorphism of the quasi-local $C^*$-algebra that can be localised in any bounded spacetime region so, in particular, in a wedge region $W$, giving rise to an endomorphism of the type $III$ factor $\A(W)$ associated with $W$. The statistical dimension of $\a$ is equal to the square root of the Jones index \cite{L89}. A description of this setting from the point of view of  information entropy has been given il \cite{L97} and plays a motivation role here. Now, if the spacetime dimension is greater than 2, $\a$ has Fermi-Bose statistics and its dimension is
a positive integer, or infinite \cite{DHR}. So, also in this case, the lower bound \eqref{Lb} holds, if $\a$ is irreversible, namely if $\a$ is not a simple charge.  

\section{Appendix. Proof of Theorem \ref{th1}}
\label{appendix}

Let $\M$ be a von Neumann algebra, $L^2(\M)$ the identity $\M - \M$ bimodule as above and $\o, \f , \psi$ faithful normal states of $\M$. Set $\Delta(\f |\o) \equiv \Delta_{L^2(\M)}(\f |\o)$ and similarly for other states. 
Then the unitaries $\Delta^{it}(\f |\o)$ and $\Delta^{it}(\o |\psi)$ on $L^2(\M)$ are intertwiners
\begin{align*}
 &\Delta^{it}(\f |\o) : L^2(\M) \to {}_{\s^\f_t} L_{\s^\o_t}^2(\M)\ , \\
  &\Delta^{it}(\o |\psi) : L^2(\M) \to {}_{\s^\o_t} L_{\s^\psi_t}^2(\M) \ ,
\end{align*}
and we may consider the bimodule relative tensor product with respect to $\o$
\[
\Delta^{it}(\f |\o)\otimes \Delta^{it}(\o |\psi) : L^2(\M)\otimes L^2(\M) \to 
{}_{\s^\f_t} L_{\s^\o_t}^2(\M)\otimes {}_{\s^\o_t} L_{\s^\psi_t}^2(\M) \ .
\]
Now, 
\[
{}_{\s^\f_t} L_{\s^\o_t}^2(\M)\otimes {}_{\s^\o_t} L_{\s^\psi_t}^2(\M)
= {}_{\s^\o_{-t}\s^\f_t} L^2(\M)\otimes  L_{\s^\psi_t \s^\o_{-t}}^2(\M)
= {}_{\s^\o_{-t}\s^\f_t}  L_{\s^\psi_t \s^\o_{-t}}^2(\M)
={}_{\s^\f_t} L_{\s^\psi_t }^2(\M)
\]
with natural identifications (again, bimodule relative tensor products w.r.t. $\o$).

Thus $\Delta^{it}(\f |\o)\otimes \Delta^{it}(\o |\psi) : L^2(\M) \to 
{}_{\s^\f_t} L_{\s^\psi_t }^2(\M)$ and we first show the following.
\begin{lemma}\label{LemmaDD}
We have
\ben\label{DD}
\Delta^{it}(\f |\o)\otimes \Delta^{it}(\o |\psi)  = \Delta^{it}(\f |\psi) \ .
\een
\end{lemma}
\begin{proof}
Let's first consider the case $\f = \psi = \o$, namely we want to show that
\begin{equation}\label{DOD}
\Delta^{it}_\o\otimes \Delta^{it}_\o  = \Delta^{it}_\o  \ ,
\end{equation}
where $\Delta_\o \equiv \Delta(\o |\o)$ is the modular operator of $\o$.  By the intertwining property, 
\[
\Delta^{it}_\o\otimes \Delta^{it}_\o  = \Delta^{it}_\o z(t) \ ,
\]
with $z$ a one parameter unitary group in the center $\cal Z$ of $\M$ (the left and right actions of $\cal Z$ on $L^2(\M)$ coincide).

On the other hand, by considering the conjugate intertwiners, we have  $\overline{ \Delta^{it}_\o} = \Delta^{it}_\o$ so
\[
 \Delta^{it}_\o\otimes \Delta^{it}_\o = \overline{ \Delta^{it}_\o}\otimes\overline{ \Delta^{it}_\o} 
 = \overline{ \Delta^{it}_\o\otimes \Delta^{it}_\o} = \overline{\Delta^{it}_\o z(t)} = \overline{\Delta^{it}_\o}\,\overline{z(t)} = \Delta^{it}_\o z(-t)\ ,
\]
hence $z(t) = 1$ and \eqref{DOD} holds. 

Note now that, if $u$ is a unitary of $\M$ acting on the left on $L^2(\M)$, then $u$ is an intertwiner 
$L^2(\M) \to {}_{{\rm Ad}u}L^2(\M)$ and we have
\[
u\otimes 1 = u 
\]
and similarly $1\otimes u = u$ if $u$ acts on the right.

We prove now \eqref{DD} with $\psi = \o$. Recall that we have
\[
\Delta^{it}(\f |\o) = u_t \Delta^{it}_\o
\]
with $u_t \equiv (D\f : D\o)_t$ the Connes Radon-Nikodym cocycle in $\M$ w.r.t. $\f$ and $\o$ (see \cite{St}). Therefore
\[
\Delta^{it}(\f |\o)\otimes \Delta^{it}(\o |\o) = 
(u_t \Delta^{it}_\o) \otimes \Delta^{it}_\o =
(u_t\otimes 1) (\Delta^{it}_\o \otimes \Delta^{it}_\o ) =
u_t \Delta^{it}_\o = \Delta^{it}(\f |\o) \ ,
\]
and the special case $\psi = \o$ is proven. The general case of formula \eqref{DD} is then proven by reapplying this argument on the right to $\Delta^{it}(\o |\psi)$.
\end{proof}
{\bf Proof of Theorem \ref{th1}}. 
$(a)$: 
By definition, we have 
\begin{align*}
\Delta_\H (\f_1 | \f_2) &= d(\f_1\cdot \ell_\H^{-1}\cdot\mu) \big/d\big( \f_2\cdot r_\H^{-1}) \ ,
\\
\Delta'_\K (\f_2 | \f_3) &= d(\f_2\cdot \ell^{-1}_\K) \big/ d (\f_3\cdot r^{-1}_\K\cdot \nu) \ ,
\end{align*}
where $\mu : r_\H(\M_2)' \to \ell_\H(\M_1)$, $\nu : \ell_\K(\M_2)' \to r_\K(\M_3)$ minimal expectations. 

We may assume $\H = {}_\r L^2(\M_2)$, $\K = L^2_\th(\M_2)$ so $ r_\H(\M_2)' = \M_2$, $\ell_\K(\M_2)' = \M'_2$ on $L^2(\M_2)$ and we have
\begin{align*}
\Delta_\H (\f_1 | \f_2) = \Delta(\f_1\cdot \ell_\H^{-1}\cdot\mu|  \f_2\cdot r_\H^{-1}) \ ,
\\
\Delta'_\K (\f_2 | \f_3) = \Delta(\f_2\cdot \ell_\K^{-1}|  \f_3\cdot r_\K^{-1}\cdot\nu) \ ,
\end{align*}
with $\Delta \equiv \Delta_{L^2(\M_2)}$ as before. 

By  Kosaki's formula \eqref{Ko}, we have
\[
\Delta'_\K (\f_2 | \f_3) =\Ind(\K)\,\Delta_\K (\f_2 | \f_3) \ ,
\]
therefore, by Lemma \ref{LemmaDD}, we get
\begin{align*}
\Delta_\H (\f_1 | \f_2)\otimes \Delta_\K (\f_2 | \f_3) &=
\Ind(\K)^{-1}\, \Delta(\f_1\cdot \ell_\H^{-1}\cdot\mu | \f_2\cdot r_\H^{-1}) \otimes \Delta(\f_2\cdot \ell^{-1}_\K | \f_3\cdot r^{-1}_\K \cdot \nu)\\
&=\Ind(\K)^{-1}\, \Delta(\f_1\cdot \ell_\H^{-1}\cdot\mu |  \f_3\cdot r^{-1}_\K \cdot \nu) \\ &=
\Delta(\f_1\cdot \ell_\H^{-1}\cdot\mu\cdot \nu' |  \f_3\cdot r^{-1}_\K ) \\ &=
\Delta_{\H\otimes \K} (\f_1 | \f_3)
\end{align*}
(where $\nu'$ is the expectation dual to $\nu$ on $L^2(\M_2)$) 
because $\Ind(\nu) = \Ind(\K)$, $\ell_\K$ and $r_\H$ are here the standard left and right actions on $L^2(\M_2)$ and $\mu\cdot \nu'$ is the minimal expectation $r_{\H\otimes\K}(\M_3)'\to  \ell_{\H\otimes\K}(\M_1)$ (multiplicativity of the minimal index \cite{KL, L92}). 

$(b)$: 
With $\e : r_\H(\M_2)' \to \ell_\H(\M_1)$ the minimal expectation,
we have by \eqref{Ko}:
\begin{multline}
\overline{\Delta_\H(\f_1 | \f_2)} = \overline{d (\f_1\cdot \ell_\H\cdot \e) \big/ d (\f_2\cdot r_\H)}
= d (\f_1\cdot r_{\bar\H}\cdot \e) \big/ d (\f_2\cdot \ell_{\bar \H}) \\
= \big( d (\f_2\cdot \ell_{\bar \H}) \big/ d (\f_1\cdot r_{\bar\H}\cdot\e)\big)^{-1}
= \Delta'_{\bar \H}(\f_2 | \f_1)^{-1} 
= \Ind(\H)^{-1}\cdot \Delta_{\bar \H}(\f_2 | \f_1)^{-1} \ ,
\end{multline}
which is equivalent to $(b)$. 

$(c)$: If $T$ is unitary, then 
\begin{multline*}
T\Delta_\H(\f_1 | \f_2)T^* = Td (\f_1\cdot \ell_\H\cdot \e) \big/ d (\f_2\cdot r_\H)T^* \\
= d (\f_1\cdot \ell_\H\cdot \e\cdot \Ad T) \big/ d (\f_2\cdot r_\H\cdot \Ad T) 
= d (\f_1\cdot \ell_\H\cdot\Ad T\cdot \e' ) \big/ d (\f_2\cdot r_\H\cdot \Ad T) \\
= d (\f_1\cdot \ell_{\H'}\cdot \e') \big/ d (\f_2\cdot r_{\H'}) 
= \Delta_{\H'}(\f_1 | \f_2)
\end{multline*}
as $\Ad T$ transforms the minimal expectation $\e: r_\H(\M_2)'\to \ell_\H(\M_1)$ to the minimal expectation $\e': r_{\H'}(\M_2)'\to \ell_{\H'}(\M_1)$, so $(c)$ holds. 

In general, consider the polar decomposition $T = vh$. Then $h$ is an intertwiner $\H\to\H$, namely $h\in \ell_\H(\M_1)'\cap r_\H(\M_2)'$, so $(c)$ holds for $h$ by \eqref{frc}. 

Now, $v$ gives a unitary intertwiner $p\H \to q\H' $ with $p=v^*v$ and $q = vv^*$, so we are left to check $(c)$ in the case $T$ is a projection  $T=e \in\ell_\H(\M_1)'\cap r_\H(\M_2)'$, namely that
\ben\label{res}
\Delta_\H(\f_1 | \f_2)|_{e\H} = \frac{d_{e\H}}{d_\H}\, \Delta_{e\H}(\f_1 | \f_2) \ ,
\een
that readily follows from \cite[Prop. 4.3]{K86}. 
\endproof

\section{Final comments}\label{FC}
We end up with a few comments.

The vacuum geometric modular action in QFT is related with the Hawking-Unruh effect, see \cite{Sew}  for the Schwarzschild black hole case. The evolution parameter of the modular group is proportional to the proper time of the geodesic observer. Our work here gives, in particular,  a further viewpoint concerning this evolution in a charged state  \cite{L97, L01}. 

One may read our present paper also in relation with the intrinsic, modular interpretation of time proposed by Connes and Rovelli in \cite{CR}. 

One would also understand how our results fit with the recent discussion of
Landauer's principle within the $C^*$-algebraic context for quantum statistical models describing a finite level quantum system coupled to an infinitely extended thermal reservoir \cite{JP}.  Under  natural assumptions, Landauer's bound saturates there. 

As mentioned, Jones' index is related to entropy \cite{PP,L97}; indeed it appears in quantum information contexts, see \cite{N13,FNO}.

One may wonder about possible relations with other forms of entropy in Quantum Field Theory. We refer to \cite{OT,HS} for recent Operator Algebraic analyses on entanglement entropy in QFT.   

The very recent paper \cite{LXre} provides a rigorous computation of the mutual information in the setting of free fermions on the circle.

\section{Outlook}
Our work is going to be naturally supplemented by two forthcoming papers. 

The paper \cite{GL} concerns the mathematical methods underlying our analysis, pointing to a clarification about the notion of dimension for bimodules over von Neumann algebras with non trivial, finite dimensional centers (for related analysis in this context, see \cite{GHJ,FI} and refs. therein). In particular, it will discuss the functoriality properties of the matrix dimension. The study of the non trivial center case is motivated in order to deal with a general quantum systems with a non trivial classical component too. Part of the analysis extends to the infinite dimensional center case, that might be discussed in the future. 

The paper \cite{LX} is going to discuss the Bekenstein bound in the context of black hole information theory. The arguments and  results in \cite{L97}, together with a due interpretation, naturally give us a rigorous derivation of this bound. Our formula for the incremental free energy can be read in this framework. A quite similar discussion can be found in a more recent paper \cite{Ca}, which is however heuristic (as local von Neumann algebras are of type $III$).
\bigskip

\noindent
{\bf Acknowledgements.} 
The author would like to thank the Isaac Newton Institute for Mathematical Sciences in Cambridge, and the Simons Foundation, for the hospitality during the program ``Operator algebras: subfactors and their applications'', supported by EPSRC Grant Number EP/K032208/1, in January-February and May-June 2017, where part of this work was carried on and presented at the June workshop.
Thanks also to Luca Giorgetti and the referee for several careful comments. 

\end{document}